\documentclass[a4paper,english,autoref,cleveref, thm-restate]{lipics-v2019}

\nolinenumbers

\title{The complexity of bounded context switching with dynamic thread creation}

\author{Pascal Baumann}{Max Planck Institute for Software Systems (MPI-SWS), Germany}{pbaumann@mpi-sws.org}{https://orcid.org/0000-0002-9371-0807}{}

\author{Rupak Majumdar}{Max Planck Institute for Software Systems (MPI-SWS), Germany}{rupak@mpi-sws.org}{https://orcid.org/0000-0003-2136-0542}{}

\author{Ramanathan S. Thinniyam}{Max Planck Institute for Software Systems (MPI-SWS), Germany}{thinniyam@mpi-sws.org}{https://orcid.org/0000-0002-9926-0931}{}

\author{Georg Zetzsche}{Max Planck Institute for Software Systems (MPI-SWS), Germany}{georg@mpi-sws.org}{https://orcid.org/0000-0002-6421-4388}{}

\authorrunning{P. Baumann, R. Majumdar, R.\,S. Thinniyam, and G. Zetzsche}

\Copyright{Pascal Baumann and Rupak Majumdar and Ramanathan S. Thinniyam and Georg Zetzsche}



\category{Track B: Automata, Logic, Semantics, and Theory of Programming}

\funding{This research was funded in part by the Deutsche Forschungsgemeinschaft project 389792660-TRR 248 (see \url{https://perspicuous-computing.science})
and by the European Research Council under the Grant Agreement 610150 (ERC Synergy Grant ImPACT).}

\acknowledgements{We would like to thank Ranko Lazi\'{c}, who made us aware of ``counter systems with chained counters'' from \cite{DBLP:conf/lics/DemriFP13} after submission of this work.}

\EventEditors{Artur Czumaj, Anuj Dawar, and Emanuela Merelli}
\EventNoEds{3}
\EventLongTitle{47th International Colloquium on Automata, Languages, and Programming (ICALP 2020)}
\EventShortTitle{ICALP 2020}
\EventAcronym{ICALP}
\EventYear{2020}
\EventDate{July 8--11, 2020}
\EventLocation{Saarbrücken, Germany (virtual conference)}
\EventLogo{}
\SeriesVolume{168}
\ArticleNo{111}

\usepackage{amsfonts}
\usepackage{mathrsfs}
\usepackage{amsmath,amssymb}
\usepackage{mathtools}

\usepackage{color}
\usepackage{xparse}
\usepackage{ifthen}

\usepackage{complexity}

\usepackage{todonotes}

\DeclareDocumentCommand{\DCPS}{O{}}{\mathsf{DCPS}\ifthenelse{\equal{#1}{}}{}{[#1]}}
\DeclareDocumentCommand{\SRP}{O{}}{\mathsf{SRP}\ifthenelse{\equal{#1}{}}{}{[#1]}  }
\DeclareDocumentCommand{\HP}{O{}}{\mathsf{HP}\ifthenelse{\equal{#1}{}}{}{[#1]}}

\DeclareDocumentCommand{\NTERM}{O{}}{\mathsf{NTERM}\ifthenelse{
\equal{#1}{}}{}{[#1]}  }
\DeclareDocumentCommand{\UBOUND}{O{}}{\mathsf{UBOUND}\ifthenelse{
\equal{#1}{}}{}{[#1]}  }
\DeclareDocumentCommand{\NFAIR}{O{}}{\mathsf{NFAIR}\ifthenelse{
\equal{#1}{}}{}{[#1]}  }
\DeclareDocumentCommand{\NSTARV}{O{}}{\mathsf{NSTARV}\ifthenelse{
\equal{#1}{}}{}{[#1]}  }

\newclass{\TWOEXPSPACE}{2EXPSPACE}

\usepackage{tikz}
\usetikzlibrary{petri}
\tikzstyle{every place}=[minimum size=5mm]
\tikzstyle{every transition}=[minimum size=5mm]
\usetikzlibrary{backgrounds,calc}

\def\qed{\rule{0.4em}{1.4ex}}

\def\@envspa{\hspace{0.3em}}
\def\@sa{\hspace{-0.2em}}
\def\@sb{\hspace{0.5em}}
\def\@sc{\hspace{-0.1em}}

{\bfseries\upshape}{\itshape}
{\bfseries\upshape}{\itshape}

\def\set#1{{\{ #1 \}}}
\def\tuple#1{{\langle #1 \rangle }}
\def\multi#1{{[\![ #1 ]\!]}}
\def\nats{{\mathbb{N}}}

\def\mmap{\mathbf{m}}

\def\card#1{\lvert {#1} \rvert}

\newcommand{\multiset}[1]{{\mathbb{M}[ #1 ]}}

\def\cN{\mathcal{N}}

\def\prod{\mathcal{P}}

\newcommand{\cA}{\mathcal{A}}

\newcommand{\rnp}{\mathsf{RNP}}
\newcommand{\TDPN}{\mathsf{TDPN}}

\DeclareDocumentCommand{\langof}{O{} m}{%
  \mathsf{L}_{#1}(#2)%
  }
\DeclareDocumentCommand{\autstep}{O{}}{%
        \xrightarrow{#1}%
        }

\DeclareDocumentCommand{\autsteps}{O{}}{%
        \xRightarrow{#1}%
      }
\DeclareDocumentCommand{\autsteph}{O{} m}{%
        \xrightarrow{#1}_{#2}%
        }

\DeclareDocumentCommand{\autstepsh}{O{} m}{%
        \xRightarrow{#1}_{#2}%
      }

\mathchardef\mhyphen="2D
\DeclareDocumentCommand{\EXPTIME}{O{}}{%
  \ifthenelse{\equal{#1}{}}{%
    \mathsf{EXPTIME}%
  }{%
    {#1}\mhyphen\mathsf{EXPTIME}%
  }%
}

\begin{document}

\maketitle

\begin{abstract}
Dynamic networks of concurrent pushdown systems ($\DCPS$) are a theoretical
model for multi-threaded recursive programs with shared global state and dynamical creation of threads.
The (global) state reachability problem for $\DCPS$ is undecidable in general, but Atig et al.\ (2009) showed that
it becomes decidable, and is in $\TWOEXPSPACE$, when each thread is restricted to a fixed number of context switches.
The best known lower bound for the problem is $\EXPSPACE$-hard and this lower bound follows already
when each thread is a finite-state machine and runs atomically to completion (i.e., does not switch contexts).
In this paper, we close the gap by showing that state reachability is $\TWOEXPSPACE$-hard already with only one
context switch.
Interestingly, state reachability analysis is in $\EXPSPACE$ both for pushdown threads without context switches 
as well as for finite-state threads with arbitrary context switches.
Thus, recursive threads together with a single context switch provide an exponential advantage.

Our proof techniques are of independent interest for $\TWOEXPSPACE$-hardness results. 
We introduce \emph{transducer-defined} Petri nets, a succinct
representation for Petri nets,
and show coverability is $\TWOEXPSPACE$-hard for this model.
To show $\TWOEXPSPACE$-hardness, we present a modified version of Lipton's simulation of counter machines
by Petri nets, where the net programs can make explicit recursive procedure calls up to a bounded depth. 
\end{abstract}

\newpage

\section{Introduction}
\label{sec:intro}

There is a complexity gap between $\EXPSPACE$ and $\TWOEXPSPACE$ that
shows up in several problems in the safety verification of multithreaded programs.

Atig, Bouajjani, and Qadeer \cite{AtigBQ2009}
study safety verification for \emph{dynamic networks of concurrent pushdown systems} ($\DCPS$),
a theoretical model for multithreaded recursive programs with a finite shared global state, 
where threads can be recursive and can dynamically spawn additional threads.
Unrestricted reachability is undecidable in this model.
To ensure decidability, like many other works \cite{QR05,LalReps,MusuvathiQadeer,LaTorre},
they assume a bound $K$ that restricts each thread to have at most 
$K$ context switches.
For safety verification in this model, formulated as global state reachability,
they show a lower bound of $\EXPSPACE$ and an upper bound of
$\TWOEXPSPACE$, ``closing the gap'' is left open.

Kaiser, Kroening, and Wahl \cite{kaiser2010dynamic} study safety verification of
multithreaded \emph{non-recursive} programs with local and global Boolean variables.
In this model, an arbitrary number of non-recursive threads execute over shared global state, but each
thread can maintain local state in Boolean variables.
Although their paper does not provide an explicit complexity bound, a lower bound of $\EXPSPACE$ and an
upper bound of $\TWOEXPSPACE$ can be derived from a reduction from
Petri net coverability and their algorithm respectively.

Interestingly, when we restrict the models to  disallow either context
switches (i.e., each thread runs atomically
to completion) or local state in the form of the pushdown stack or local variables (but allow arbitrary
context switches),
safety verification is in $\EXPSPACE$~\cite{AtigBQ2009,GantyM12}.

Thus, the complexity gap asks whether or not the combination of \emph{local state} (maintained in local variables or in the stack)
and bounded \emph{context switching} provides additional power to computation.
In this paper, we show that indeed it does. In fact, the combination
of local state and just \emph{one} context
switch is sufficient to achieve $\TWOEXPSPACE$ lower bounds for these problems.
This closes the complexity gap.

We believe the constructions and models that we use along the way are of independent interest.
We introduce \emph{transducer-defined} Petri nets ($\TDPN$s), a succinct representation for
Petri nets.
The places in a $\TDPN$ are encoded using words over a fixed alphabet, and the transitions
are described by length-preserving transducers.
We show that coverability for $\TDPN$s is $\TWOEXPSPACE$-complete\footnote{After submitting this work, the authors were made aware of ``(level~1) counter systems with chained counters'' from \cite{DBLP:conf/lics/DemriFP13}, for which $\TWOEXPSPACE$-hardness of state reachability is shown in~\cite[Theorem 14]{DBLP:conf/lics/DemriFP13}. The $\TWOEXPSPACE$-hardness of coverability in $\TDPN$ could also be deduced from that result.} and
give a polynomial-time reduction from coverability for $\TDPN$s to
safety verification for $\DCPS$ with one context switch.

The idea of the latter reduction is to map a (compressed) place to the stack of
a thread and a marking to the set of currently spawned threads.
A key obstacle in the simulation is to ``transfer'' potentially exponential amount
of information from before a transition to after it through a polynomial-sized global store.
We present a ``guess and verify'' procedure, using non-determinism and
the use of additional threads to verify a stack content letter-by-letter.

In order to show $\TWOEXPSPACE$-hardness for $\TDPN$s, we introduce the
model of \emph{recursive net programs} ($\RNP$s), which add the power of making possibly recursive
procedure calls to
the model of net programs (i.e.,  programs with access to Petri net counters). 
The addition of recursion enables us to replace the ``copy and paste code''
idea in Lipton's construction to show $\EXPSPACE$-hardness of Petri net coverability
\cite{lipton1976reachability} 
with a more succinct and cleaner program description where the copies are instead
represented by different values of the local variables of the procedures.
The net effect is to push the requirement for copies into the
call stack of the $\rnp$ while maintaining a syntax which gives us a
$\rnp$ which is polynomial in the size of a given counter program.
When the stack size is bounded by an exponential function of the size of the program,
we get a $\TWOEXPSPACE$-lower bound.
We show that recursive net programs with exponentially large stacks
can be simulated by $\TDPN$s.

Finally, we note that the $\TWOEXPSPACE$ lower bound holds for $\DCPS$ where each stack
is bounded by a linear function of the size.
Such stacks can be encoded by polynomially many local Boolean
variables, giving us a $\TWOEXPSPACE$ lower bound for the model of
Kaiser et al.

In summary, we introduce a number of natural $\TWOEXPSPACE$-complete problems and,
through a series of reductions, close an exponential gap
in the complexity of safety verification for multithreaded recursive programs.

\section{Dynamic Networks of Concurrent Pushdown Systems ($\DCPS$)} \label{sec:dcps}

In this section, we define the model of $\DCPS$ and then state our main result.
Intuitively, a $\DCPS$ consists of a finite state control and several pushdown threads with local configurations, one 
of them being the active thread.
A local configuration contains the number of context switches the thread has already performed, as well as the contents of its local stack. 
An action of a thread may specify a new thread with initially one symbol on the stack to be spawned as an inactive thread. 
The active thread can be switched out for one of the inactive threads at any time. 
When a thread is switched out, its context switch number increases by one.
One can view this model as a collection of dynamically created recursive threads (with a call stack each), that communicate using some finite shared memory (the state control).

A \emph{multiset} $\mmap\colon S\rightarrow\nats$ over a set $S$ maps each
element of $S$ to a natural number.
Let $\multiset{S}$ be the set of all multisets over $S$.
We treat sets as a special case of multisets 
where each element is mapped onto $0$ or $1$.
We sometimes write
$\mmap=\multi{a_1,a_1,a_3}$ for the multiset
$\mmap\in\multiset{S}$ such that $\mmap(a_1)=2$, $\mmap(a_3)=1$, and $\mmap(a) = 0$ for each $a \in S\backslash\set{a_1,a_3}$. 
The empty multiset is denoted \(\emptyset\).
The \emph{size} of a multiset $\mmap$, denoted $\card{\mmap}$, is
given by $\sum_{a\in S}\mmap(a)$.
Note that this definition applies to sets as well.
 
Given two multisets $\mmap,\mmap'\in\multiset{S}$ we define $\mmap\oplus
\mmap'\in\multiset{S}$ to be a multiset such that for all $a\in S$,
we have $(\mmap\oplus \mmap')(a)=\mmap(a)+\mmap'(a)$.
We also define the natural order
$\preceq$ on $\multiset{S}$ as follows: $\mmap\preceq\mmap'$ if{}f there
exists $\mmap^{\Delta}\in\multiset{S}$ such that
$\mmap\oplus\mmap^{\Delta}=\mmap'$. We also define $\mmap \ominus
\mmap'$ for $\mmap' \preceq \mmap$ analogously: for all $a\in S$,
we have $(\mmap\ominus \mmap')(a)=\mmap(a)-\mmap'(a)$.

A \emph{Dynamic Network of Concurrent Pushdown Systems ($\DCPS$)} $
\mathcal{A} = (G,\Gamma,\Delta,g_0,\gamma_0)$ consists of a finite set
of \emph{(global) states} $G$, a finite alphabet of \emph{stack
symbols} $\Gamma$, an \emph{initial state} $g_0 \in G$, an 
\emph{initial stack symbol} $\gamma_0 \in \Gamma$, and a finite set of
\emph{transition rules} $\Delta$. Elements of $\Delta$ have one of
the two forms (1) $g|\gamma \hookrightarrow g'|w'$, or (2) $g|\gamma
\hookrightarrow g'|w' \triangleright \gamma'$, where $g,g' \in G$,
$\gamma,\gamma' \in \Gamma$, $w' \in \Gamma^*$, and $|w'| \leq 2$. 
Rules of the first kind allow the $\DCPS$ to take a single step
in one of the pushdown threads while the second additionally
spawn a new thread with top of stack symbol $\gamma'$.
The \emph{size} of $\mathcal{A}$ is defined as $|\mathcal{A}| = |G| + |\Gamma| + |\Delta|$.
 
The set of \emph{configurations} of $\mathcal{A}$ is $G \times (\Gamma^*
\times \nats) \times \multiset{\Gamma^* \times \nats}$. Given a
configuration $\langle g, (w,i), \mmap \rangle$, we call $g$ the 
\emph{(global) state}, $(w,i)$ the \emph{local configuration} of the 
\emph{active thread}, and $\mmap$ the multiset of the \emph{local
configurations} of the \emph{inactive threads}. The initial
configuration of $\mathcal{A}$ is $\langle g_0, (\gamma_0,0),
\emptyset \rangle$. For a configuration $c$ of $\cA$, we will
sometimes write $c.g$ for the state of $c$ and $c.\mmap$ for the
multiset
of threads of $c$ (both active and inactive).
The \emph{size} of a configuration $c = \langle g, (w,i), \mmap \rangle$
is defined as $|c| = |w| + \sum_{(w',j) \in \mmap}|w'|$.

For $i \in \nats$ we define the relation $\Rightarrow_i = \rightarrow_i \cup \mapsto_i$ on configurations of $\mathcal{A}$, where $\rightarrow_i$ and $\mapsto_i$ are defined as follows:
\begin{itemize}
  \item $\langle g, (\gamma.w,i), \mmap \rangle \rightarrow_i \langle g', (w'.w,i), \mmap' \rangle$ for all $w \in \Gamma^*$ iff (1) there is a rule $g|\gamma \hookrightarrow g'|w' \in \Delta$ and $\mmap' = \mmap$ or (2) there is a rule $g|\gamma \hookrightarrow g'|w' \triangleright \gamma' \in \Delta$ and $\mmap' = \mmap \oplus \multi{(\gamma',0)}$.
  \item $\langle g, (w,i), \mmap \oplus \multi{(w',j)} \rangle \mapsto_i \langle g, (w',j), \mmap \oplus \multi{(w,i+1)} \rangle$ for all $j \in \nats,$ $g \in G,$\linebreak$\mmap \in \multiset{\Gamma^* \times \nats}$, and $w,w' \in \Gamma^*$.
\end{itemize} 
For $b \in \nats$ we define the relation $\Rightarrow_{\leq b} :=
\bigcup_{i=0}^b \Rightarrow_i$. We use $\Rightarrow_i^*$ and
$\Rightarrow_{\leq b}^*$ to denote the reflexive, transitive closure
of $\Rightarrow_i$ and $\Rightarrow_{\leq b}$, respectively. 

Given $K \in \nats$, a state $g$ of $\mathcal{A}$ is \emph{$K$-bounded reachable} iff $\langle g_0, (\gamma_0,0), \emptyset \rangle \Rightarrow_{\leq K}^* \langle g, (w,i), \mmap \rangle$ for some $(w,i) \in \Gamma^* \times \{0,\ldots,K\}$ and $\mmap \in \multiset{\Gamma^* \times \{0,\ldots,K+1\}}$.

Intuitively, a local configuration $(w,i)$ describes a pushdown thread with stack content $w$ that has already performed $i$ context switches. 
The relation~$\rightarrow_i$ corresponds to applying the two kinds of transition rules at $i$ context switches. 
Both of them define pushdown transitions, which the active thread can perform. 
Type (2) also spawns a new inactive pushdown thread with $0$ context switches, whose initial stack content consists of a single specified symbol. 
For each $i \in \nats$, the relation~$\mapsto_i$ corresponds to switching out the active thread and raising its number of context 
switches from $i$ to $i+1$, while also switching in a previously inactive thread. 
For a fixed $K$, the \emph{$K$-bounded state
reachability problem} ($\SRP[K]$) for a $\DCPS$ is :
\begin{description}
  \item[Input] A $\DCPS$ $\mathcal{A}$ and a global state
  $g$
  \item[Question] Is $g$ $K$-bounded reachable in $\mathcal{A}$?
\end{description}
This corresponds to asking whether the global state $g$ is reachable if each thread can perform at most $K$ context switches.

\begin{theorem}[Main Result]
\label{thm:srp_dcps}
  For each $K\ge 1$, the problem $\SRP[K]$ is $\TWOEXPSPACE$-complete.
\end{theorem}

The fact that $\SRP[K]$ is in $\TWOEXPSPACE$ for any fixed $K$
follows from the results of Atig et al.~\cite{AtigBQ2009}.
They use a slightly different variant of $\DCPS$. However, it is possible to show a reduction from $\SRP[K]$ for our variant to $\SRP[K+2]$ for theirs, see Appendix~\ref{sec:inheritance}.

Our main result is to show $\TWOEXPSPACE$-hardness for $\SRP[1]$. One may also
adapt the results of Atig et al.\ to the problem where $K$ is part of
the input (encoded in unary),
to derive an $\EXPSPACE$ lower bound and a $\TWOEXPSPACE$ upper bound.
Our result immediately implies $\TWOEXPSPACE$-hardness for this problem as well. 

In the remaining sections  we prove the lower bound in Theorem~\ref{thm:srp_dcps}.
In Section~\ref{sec:tdpn_to_dcps}, we introduce \emph{transducer-defined Petri nets} ($\TDPN$),
a succinct representation for Petri nets for which we prove the coverability problem
is $\TWOEXPSPACE$-complete.
Then, we show a reduction from the coverability problem for $\TDPN$s to the $\SRP[1]$ problem. 
In Section~\ref{sec:lower_bound}, we prove hardness for coverability of $\TDPN$s,
completing the proof.


\def \transpn {\mathcal{N}}
\def \tmove {\mathcal{T}_{\mathit{move}}}
\def \tjoin {\mathcal{T}_{\mathit{join}}}
\def \tfork {\mathcal{T}_{\mathit{fork}}}

\def \qmove {Q_{\mathit{move}}}
\def \qjoin {Q_{\mathit{join}}}
\def \qfork {Q_{\mathit{fork}}}

\def \dmove {\Delta_{\mathit{move}}}
\def \djoin {\Delta_{\mathit{join}}}
\def \dfork {\Delta_{\mathit{fork}}}

\def \initmove {q_{\mathit{move0}}}
\def \initjoin {q_{\mathit{join0}}}
\def \initfork {q_{\mathit{fork0}}}

\def \finmove {Q_{\mathit{movef}}}
\def \finjoin {Q_{\mathit{joinf}}}
\def \finfork {Q_{\mathit{forkf}}}

\def \pda {\mathcal{P}}
\def \main {\mathit{main}}
\def \check {\mathit{check}}
\def \halt {\mathit{halt}}
\def \kill {\mathit{kill}}
\def \return {\mathit{return}}

\section{ Transducer Defined Petri Nets ($\TDPN$)}
\label{sec:tdpn_to_dcps}

In this section, we prove the lower bound in Theorem 
\ref{thm:srp_dcps} by reducing 
coverability for a succinct representation of Petri nets, namely
$\TDPN$, to $\SRP[1]$ for $\DCPS$. We first recall some definitions
about Petri nets,
transducers
and problems
related to them.

\begin{definition}
\label{def:petri_net} 
  A \textbf{Petri net} is a tuple $N = (P,T,F,p_0,p_f)$ where
  $P$ is a finite
  set of \emph{places}, $T$ is a finite set of \emph{transitions} with
  $T \cap P = \varnothing$, $F \subseteq (P \times T) \cup (T \times
  P)$ is its \emph{flow relation}, and $p_0 \in P$ (resp. $p_f \in P$)
  its
  \emph{initial place} (resp. \emph{final place}). 
  A \emph{marking} of $N$ is a multiset $\mmap\in \multiset{P}$. 
  For a marking $\mmap$ and a place $p$ we say that there
  are $\mmap(p)$ \emph{tokens} on $p$. Corresponding to the initial (resp.\ final) 
  place we have the initial marking $\mmap_0=\multi{p_0}$ 
  (resp.\ final marking $\mmap_f=\multi{p_f}$).
  The \emph{size} of $N$ is defined as $|N| = |P|+|T|$.

  A transition $t \in T$ is enabled at a marking $\mmap$ if $ \set{p \mid (p,t) \in F} \preceq \mmap $. If $t$ is enabled in $\mmap$, $t$ can be fired, which leads to a marking $\mmap'$ with $\mmap' = \mmap \oplus \set{p \mid (t,p) \in F} \ominus \set{p \mid (p,t) \in F}$. In this case we write $\mmap \autstep[t] \mmap'$.  
  A marking $\mmap$ is \emph{coverable} in $N$ if there is a sequence
  $\mmap_0 \autstep[t_1] \mmap_1 \autstep[t_2] \ldots \autstep[t_l] \mmap_l$ such that $\mmap \preceq \mmap_l$. We call such a sequence a \emph{run} of $N$.
\end{definition}
The \emph{coverability problem} for Petri nets is defined as:
\begin{description}
	\item[Input] A Petri net $N$.
	\item[Question] Is $\mmap_f$ coverable in $N$?
\end{description}

\begin{definition}
\label{def:transducer}
  For $n \in \nats$, a \emph{(length preserving)} \textbf{$
  \boldsymbol{n}$-ary transducer} $\mathcal{T} = (\Sigma,Q,q_0,Q_f,\Delta)$ consists of an alphabet $\Sigma$, a finite set of \emph{states} $Q$, an \emph{initial state} $q_0 \in Q$, a set of \emph{final states} $Q_f \subseteq Q$, and a \emph{transition relation} $\Delta \subseteq Q \times \Sigma^n \times Q$. For a \emph{transition} $(q,a_1,\ldots,a_n,q') \in \Delta$ we also write $q \xrightarrow{(a_1,\ldots,a_n)} q'$.
  The \emph{size} of $\mathcal{T}$ is defined as $|\mathcal{T}| = n\cdot|\Delta|$.
  
  The \emph{language} of $\mathcal{T}$ is the $n$-ary relation $L(\mathcal{T}) \subseteq (\Sigma^*)^n$ containing precisely those $n$-tuples $(w_1,\ldots,w_n)$, for which there is a transition sequence
  $q_0 \xrightarrow{(a_{1,1},\ldots,a_{n,1})} q_1 \xrightarrow{(a_{1,2},\ldots,a_{n,2})} \ldots \xrightarrow{(a_{1,m},\ldots,a_{n,m})} q_m$
  with $q_m \in Q_f$ and $w_i = a_{i,1}a_{i,2} \cdots a_{i,m}$ for
  all $i \in \{1, \ldots, n\}$. Such a transition sequence is called an
  \emph{accepting run} of $\mathcal{T}$.
  
\end{definition}  
We note that in the more general (i.e. non-length-preserving)
definition of a transducer, the transition relation $\Delta$ is a subset
of
$Q \times (\Sigma \cup \varepsilon)^n \times Q$. All transducers we
consider in this paper are length-preserving.  

\begin{definition}
  \label{def:tdpn}
  A \textbf{transducer-defined Petri net} $\mathcal{N} = (w_{
  \mathit{init}},$ $w_{
    \mathit{final}},$ $
  \mathcal{T}_{move},$ $\mathcal{T}_{fork},$ $\mathcal{T}_{join})$
  consists of two words $w_{\mathit{init}},w_{
  \mathit{final}} \in \Sigma^l$ for some $l \in \nats$, a binary
  transducer $
  \mathcal{T}_{move}$ and two ternary transducers $\mathcal{T}_
  {fork}$ and $\mathcal{T}_{join}$. Additionally, all three
  transducers share $\Sigma$ as their alphabet. This
  defines an \emph{explicit} Petri net $N(\cN) = 
  (P,T,F,p_0,p_f)$ :
  \begin{itemize}
  \item $P=\Sigma^l$. 
  \item $T$ is the disjoint union of $T_{\mathit{move
  }}$, $T_{\mathit{join}}$ and $T_{\mathit{fork}}$\footnote{Note that a tuple $
  (w,w',w'') \in T_{\mathit{join}}$ is different from the
  same tuple in $T_{\mathit{fork}}$. In the interest of readability,
  we have chosen not to introduce
  a $4^{th}$ coordinate to distinguish the two.} where
  \begin{itemize}
    \item $T_{\mathit{move}} =\{(w,w') \in \Sigma^l \times \Sigma^l \mid 
        (w,w') \in L(\mathcal{T}_{\mathit{move}}) \}$,
    \item $T_{\mathit{fork}} =\{(w,w',w'') \in \Sigma^l \times \Sigma^l \times \Sigma^l \mid
        (w,w',w'') \in L(\mathcal{T}_{\mathit{fork}}) \}$, and
     \item $T_{
         \mathit{join}}= \{(w,w',w'') \in \Sigma^l \times \Sigma^l \times
         \Sigma^l \mid
         (w,w',w'') \in L(\mathcal{T}_{\mathit{join}}) \}$.
  \end{itemize}
  \item $ p_0 = w_{\mathit{init}}$ and $p_f=w_{\mathit{final}}$.
  \item $\forall t \in T\colon$
    \begin{itemize}
    \item If $t = (p_1,p_2) \in T_{\mathit{move}}$ then $(p_1,t),
    (t,p_2) \in
    F$.
    \item If $t = (p_1,p_2,p_3) \in T_{\mathit{fork}}$ then $(p_1,t),
    (t,p_2),
    (t,p_3) \in F$.
    \item If $t = (p_1,p_2,p_3) \in T_{\mathit{join}}$ then $(p_1,t),
    (p_2,t),
    (t,p_3) \in F$.
    \end{itemize}
  \end{itemize}
  An accepting run of one of the transducers, which corresponds to
  a single transition of $N$, is called a \emph{transducer-move}.
  The \emph{size} of $\mathcal{N}$ is defined as $|\mathcal{N}| = l + |\mathcal{T}_{move}| + |\mathcal{T}_{fork}| + |\mathcal{T}_{join}|$.
\end{definition}

\begin{figure}[t]
  \def \dist{.4}
  \centering
  \begin{tikzpicture}[]
    \node[place,label=left:$p_1$] (p1) at (0,0) {};
    \node[place,label=right:$p_2$] (p2) at (2,0) {};
    \node[transition,label=below:$t$] at (1,0) {}
      edge[pre] (p1)
      edge[post] (p2);
    \node at (1, -1.2) {$t = (p_1,p_2) \in L(\mathcal{T}_{move})$};
    
    \node[place,label=left:$p_1$] (q1) at (4.5,0) {};
    \node[place,label=right:$p_2$] (q2) at (6.5,\dist) {};
    \node[place,label=right:$p_3$] (q3) at (6.5,-\dist) {};
    \node[transition,label=below:$t$] at (5.5,0) {}
      edge[pre] (q1)
      edge[post] (q2)
      edge[post] (q3);
    \node at (5.5, -1.2) {$t = (p_1,p_2,p_3) \in L(\mathcal{T}_{fork})$};
    
    \node[place,label=left:$p_1$] (r1) at (9,\dist) {};
    \node[place,label=left:$p_2$] (r2) at (9,-\dist) {};
    \node[place,label=right:$p_3$] (r3) at (11,0) {};
    \node[transition,label=below:$t$] at (10,0) {}
      edge[pre] (r1)
      edge[pre] (r2)
      edge[post] (r3);
    \node at (10, -1.2) {$t = (p_1,p_2,p_3) \in L(\mathcal{T}_{join})$};
  \end{tikzpicture}
  \caption{The types of transitions defined by the three transducers.}
  \label{fig:PNtransducers}
\end{figure}
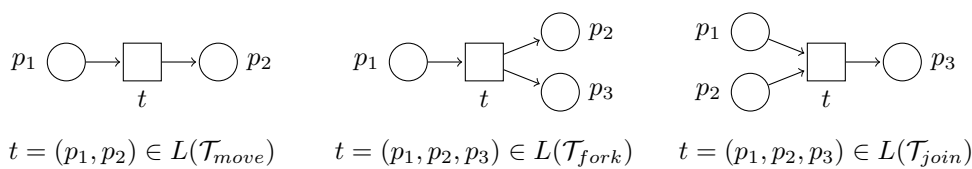

\noindent A Petri net defined by transducers in this way can only contain three different types of transitions, each type corresponding to one of the three transducers. These transition types are depicted in Figure~\ref{fig:PNtransducers}. 
The \emph{coverability problem} for $\TDPN$ is given by:
\begin{description}
	\item[Input] A $\TDPN$ $\transpn$.
  \item[Question] Is $\mmap_f=\multi{w_{\mathit{final}}}$ coverable in
  the
  corresponding explicit
  Petri net $N(\cN)$?
\end{description}

Observe that the exlicit Petri net $N(\cN)$ has $|\Sigma|^l$ places,
which is exponential in the size of $\cN$. This means that $\TDPN$ are
exponentially succinct representations of Petri nets.

It is a common theme in complexity theory to consider succinct versions of decision
problems~\cite{LEISS1981323,galperin1983succinct,papadimitriou1986note}. The resulting complexity is usually one exponent higher than the original version.
In fact, certain types of hardness proofs can be lifted generically~\cite{papadimitriou1986note}
(but such a simple argument does not seem to apply in our case).
The hardness proof in the following is deferred to Section~\ref{sec:lower_bound}.

\begin{theorem}
\label{thm:cover_tdpn}
	The coverability problem for $\TDPN$ is $\TWOEXPSPACE$-complete.
\end{theorem}
Traditionally, succinct versions of graphs and automata feature a
compression using
circuits~\cite{galperin1983succinct,papadimitriou1986note} or
formulas~\cite{LEISS1981323}.  One could also compress Petri nets by
using circuits to accept binary encodings of elements $(p,t)$ or
$(t,p)$ of the flow relation. It is relatively easy to reduce
coverability for $\TDPN$ to this model by encoding transitions $t$ as
the pair or triple of places that they correspond to, yielding
$\TWOEXPSPACE$-hardness. We consider transducers because they make the
reduction to $\DCPS$ more natural.
$\TWOEXPSPACE$-membership for any such
representation follows by first unravelling the Petri
net and then checking coverability~\cite{Rackoff78}.  

We now show that coverability
for $\TDPN$ can be
reduced in polynomial time
to $\SRP[1]$ for $\DCPS$. 
The goal of the reduction is, given a $\TDPN$ $\transpn=$ $
(w_{\mathit{init}},$
 $w_{\mathit{final}},$ $
  \mathcal{T}_{move},$ $\mathcal{T}_{fork},$ $\mathcal{T}_{join})$,
  to produce a $\DCPS$ $\cA
(\transpn)$ with a global state $\halt$ such that $w_{\mathit{final}}$ is
coverable in $\transpn$ iff $\halt$ is $1$-bounded reachable in
$\cA
(\transpn)$. 
We outline the main ideas 
and informally explain the solution to some technical issues that arise;
the formal construction is in Appendix~\ref{appendix:killDCPS}.

\subparagraph{Representation of Markings}
\label{para:rep_of_marking}
The main idea behind the simulation of a $\TDPN$ $\transpn$ by a
$\DCPS$ $\cA(\transpn)$ is
that a token on a place $w$ of $\transpn$ is represented by a thread
with stack content $w$. Extending this idea, a marking is represented by a multiset of threads, one for each
token. 

\subparagraph{Initialization}
The initial marking of $\transpn$ is $\multi{w_{\mathit{init}}}$ and $\cA
(\transpn)$ starts by going into a special state where it
always fills its stack with $w_
{\mathit{init}}$ and then moving to a global
state $\main$.
We need $O(l)$ states in the global memory for the initialization.

\subparagraph{Simulation of one Transducer-move}
In the sequel, we explain the simulation of a single 
transducer-move from $\tmove$; the changes required to be made in the
case of $\tjoin$ and $\tfork$ are explained at the end. Remembering
the choice of
transducer incurs a multiplicative cost of 3 in the global memory.
 The
transducer-move requires us to do two things: Read the stack contents
of a particular \emph{input} thread which corresponds to a place $w$
from
which a token is removed; after which we need to create an 
\emph{output} thread
which corresponds to a place $w'$ to which a token is added.
This results
in the following issue regarding input threads:
\begin{description}
 	\item[Issue 1:] How can an input thread communicate its 
 	 stack content $w$ which comes from an exponentially large space of
 	 possibilities (since this space is $\Sigma^l$) given the
   requirement for the
 	 global state space to be polynomial in $|\transpn|$?
 	 \item[Solution 1:] We pop the contents of the thread while
 	 simultaneously
 	 spawning \emph{bit-threads}, each of which contains one letter of
 	 $w$ along with the index $i \in \{ 1, \ldots, l\}$ of the letter
    and the information that
 	 $w$ is a place from which
 	 a token is being removed; all of which is coded into a single
   \emph{bit-symbol}.
 \end{description}
 Note that we have two types of
 threads:
   bit-threads and token-threads (i.e., those whose stack contents
   encode a token's position). Moreover, these two types of threads
   have disjoint sets of stack symbols: bit-symbols and 
   \emph{token-symbols}.  The idea used to solve Issue 1 and read the
   stack contents, cannot be used in reverse to create an output
   token-thread since it is not possible to populate a stack with
   information from bit-threads.
  \begin{description}
 	\item[Issue 2:] How do we ensure the creation of appropriate output
 	threads?
 	\item[Solution 2:] We implement a `guess-and-verify'
 	procedure whereby we
 	first guess the contents of an output token-thread while
   simultaneously
 	producing bit-threads corresponding to $w'$; this is followed by a
 	verification of the transition by comparing bit-threads produced
 	corresponding to $w$ and $w'$, in a bit-by-bit fashion.
 \end{description}
In particular, our simulation of a single transducer-move corresponds
to
a loop on the global state $\main$ which
is broken up into three stages: Read, guess and verify. The
implementation
of this loop ensures that a configuration $c$ of $\cA
(\transpn)$ where $c.g=\main$ has a
multiset
$c.\mmap$ of threads faithfully representing a marking $\tilde{\mmap}$
of
$\transpn$ in that $c.\mmap$ contains exactly $\tilde{\mmap}(w'')$
token-threads with
stack
content
$w''$ for each place $w''$ of $\transpn$ and no other threads. 

We note that the discussion so far shows how the run of a $\transpn$
can be simulated when the schedule switches contexts at appropriate
times.
We must also ensure that new behaviors cannot arise due to context
switches at arbitrary other points.
We accomplish this by using global \emph{locks} that ensure unwanted context switches
get stuck.
\begin{description}
  \item[Issue 3:] How do we control the effect of arbitrary context
  switches?
  \item[Solution 3:]  The global state is
  partitioned in such
  a way as to only enable operations on bit-symbols while in some
  states and token-symbols in others. We ensure that for every
  bit-symbol $\gamma$, there is at most one thread with top of stack
  $\gamma$ at any given time. Thus with
  the help of global control, we make sure
  unwanted
  context switches to bit-threads get the system stuck. The problem
  reduces to avoiding
  unwanted context switches
  between token-threads.
  
  We use a \emph{locking mechanism}.
  We add an
  extra $\top$ symbol at the top of every token-thread when it is
  first created. A read-stage always begins in a special state used for unlocking a thread 
  (i.e. removing $\top$). While reading a particular thread, the
  global state disallows any transition on $\top$ or bit-symbols. Since
  all inactive token-threads have
  $\top$ as the top of stack symbol, this implies that the system
  cannot proceed until it switches back to the unlocked token-thread.
  Similarly, during the guess-stage where we are creating a new
  token-thread, transitions are disallowed on $\top$ and bit-symbols.
  The verify-stage only operates on bit-threads and switching to a
  token-thread is similarly pointless.
\end{description}
We now describe the three stages. 
Recall that the global state keeps the information that the current
step is a transducer-move from $\tmove$.
\subparagraph{Read-stage:} We non-deterministically switch to a
token-thread $t_0$
containing
$w$ as stack content, which we need to read.
As explained earlier, we produce bit-threads decorated appropriately
and
at the end of this stage, we have popped all of $t_0$ and created $l$
bit-threads; $t_0$ ceases to exist. The number of global states
required in the stage is $O(l)$.

\subparagraph{Guess-stage:} Next,
we
create a new
token-thread with $w'$ as its stack contents by non-deterministic
guessing,
simultaneously
spawning bit-threads for each letter of $w'$. At the end of this stage $l$ more
bit-threads have been added to the task buffer (for a total of $2l$
bit-threads) along with a token-thread
containing
$w'$. As in the read-stage, the number of global states used in
this stage is $O(l)$.

\subparagraph{Verify-stage:} We guess a sequence
of transitions $\delta_1 \ldots \delta_l$  of $\tmove$ on-the-fly;
we guess $\delta_i$ which must be of the form $q_{i-1} 
\xrightarrow{(w_i,w'_i)} q_{i}$ where $w_i$ (resp. $w'_i$)
the $i^{th}$ letter of $w$ (resp. $w'$). 
We verify our guess by
comparing each $\delta_i$ with the corresponding bit-threads
$b_i,b'_i$
with index~$i$ produced
in the read-stage from $w,w'$ respectively, before moving on to $\delta_{i+1}$. During
the verification, the bit-threads
are \emph{killed}. We enforce the condition that
the target state of $\delta_i$ matches the source state of $\delta_
{i+1}$.

\subparagraph{Claim:} Killing a bit-thread $t'$ with a single stack
symbol
$\gamma'$ can be simulated by a $\DCPS$. Consider the
following sequence of operations starting from
global
state $g$ with an
active thread $t$ which contains only one symbol $\gamma$ on the stack:

\begin{enumerate}
  \item Spawn a thread $t''$ with a special symbol $\gamma_{
  \mathit{spawn}}$ and move to a special kill-state $\kill$
  which contains information regarding the state $g$ prior to the kill
  operation and stack symbols of $t$ and $t'$.
  \item Switch to a thread with symbol $\gamma'$ and pop its
  contents while moving to a special state $\return$ which is
  forwarded the information contained in $\kill$.
  \item Switch to the thread with $\gamma_{
  \mathit{spawn}}$ as top of stack and replace it with $\gamma$ and
  at the same time go to global state $g$.
\end{enumerate}
This
concludes our proof sketch of the claim. Adding a kill operation to a $\DCPS$ only incurs a polynomial
increase in the size of the $\DCPS$. A formal proof can be found in Appendix~\ref{appendix:killDCPS}.

In our setting, the net
result of the sequence of operations simulating a kill-move is to
remove the two
bit-threads $b_i,b'_i$ from
the multiset of threads without changing the
global state or the top of stack symbol $\gamma$. The special states
$\kill$ (resp. $\return$) ensure that if one switches to a thread whose
top of stack is different from $\gamma'$ in Step 2 (resp. $\gamma_{
\mathit{spawn}}$ in Step 3), no transition can be made. We return
to our
discussion regarding the sequence of transitions~$\delta_i$.

Since this
process of checking the transducer-move occurs bit-by-bit, we
require $O(l|\tmove|)$ many global states in this stage.
At the end of the verification process, $\cA(\transpn)$ is
once again in state $\main$ and the new multiset is the result of the
addition of a $w'$ thread and removal of the $w$ thread from the old
multiset of threads. 
We can now simulate the next transducer-move.

\subparagraph{Checking for Coverability}
At any point when $\cA
(\transpn)$ is in the state $\main$, it makes a non-deterministic
choice between simulating the next transducer-move or checking for
coverability. In the latter case, it goes into a special $\check$
state where the active thread is compared letter by letter with $w_
{\mathit{final}}$ in a process similar to initialization. At the end
of the checking process, $\cA
(\transpn)$ reaches the state $\halt$. 
If the check fails at any intermediate point, $\cA
(\transpn)$ terminates without reaching the $\halt$ state. 
We require a further $O(l)$ states for checking coverability.

\subparagraph{Fork and Join} We have shown above how a
single
transducer-move is
simulated assuming that it is a transducer-move from $\tmove$. In
general, the transducer-move could be from $\tjoin$ or $\tfork$ as
well. In these two cases, we have
triples of the form $(w,w',w'')$ accepted by the transducer.
However, in the former, we read $w,w'$ and guess $w''$ while
in the latter, we read $w$ and guess $w',w''$.
In the case of $\tjoin$, once we have read $w$, we
non-deterministically switch to a thread containing $w'$ as
its
contents. Whenever the threads picked
during the read-stage and the threads created during the guess-stage
do not agree with the guessed transitions of the transducer-move,
we encounter a problem during the verify-stage and $\cA
(\transpn)$ terminates without reaching the $\halt$ state.

\subparagraph{Context Switches}
Every thread (other than the initial one for $w_{\mathit{init}}$) is created
during the guess-stage and then switched out once. The next time it is
switched in, it is read and ceases to exist. This implies that there
exists a run of $\cA
(\transpn)$ simulating a run of $\transpn$ where every thread
undergoes at most one context switch. Conversely, we show that a run
of $\cA
(\transpn)$  reaching $\halt$ where every thread is bounded by at most 1
context switch
implies the existence of a run in $\transpn$ which covers the final
marking as desired.

This concludes our overview of the construction of $\cA
(\transpn)$ and completes the reduction of coverability for
$\TDPN$ to $\SRP[1]$
for $\DCPS$.
The global memory is polynomial in the size of $\transpn$.
Similarly, the stack
alphabet is expanded to include $O(l\cdot|\Sigma|)$ bit symbols,
hence the alphabet of $\cA(\transpn)$ is polynomial as well.
In summary, $\cA(\transpn)$ can be produced in time polynomial in the
size of the input. 
Details of the reduction
are in Appendix \ref{appendix:killDCPS}.

\begin{remark}
Our lower bound holds already for $\DCPS$ where the stack of each thread is
bounded by a linear function of the size of the $\DCPS$. 
Thus, as a corollary, we get $\TWOEXPSPACE$-hardness for a related model in which
each thread is a \emph{Boolean program}, i.e., where each thread has its stack bounded by a constant
but has a polynomial number (in the size of $|G|+|\Gamma|+|\Delta|$) of local Boolean variables. 
This closes the gap from \cite{kaiser2010dynamic} as well as other similar models 
studied in the literature \cite{CookKS07,Kochems14,DOsualdoKO13}.
\end{remark}

\section{Recursive Net Programs ($\rnp$)} \label{sec:lower_bound}

We prove Theorem~\ref{thm:cover_tdpn} by adapting the Lipton construction~\cite{lipton1976reachability}, as it is explained in \cite{Esp98a}, 
to our succinct representation of Petri nets. Our construction requires two steps.
First we reduce termination for bounded counter programs to
termination for \emph{Petri net} programs which do not allow zero tests. 
Second, we reduce termination of net programs with to coverability for $\TDPN$.

For the first step, we have to show how we can simulate the operation of a bounded
counter program with one without zero tests.
In the Lipton construction, this is achieved by constructing a gadget
that performs zero tests for counters bounded by some bound $B$. These
gadgets are obtained by transforming a gadget for bound $B$ into a
gadget for $B^2$. Starting with $B=2$ and applying this transformation
$n$ times leads to a gadget for $B=2^{2^n}$. One then has to argue
that the resulting net program still has linear size in the parameter $n$. 
For a $\TWOEXPSPACE$ lower bound, one would need to simulate a program
where the bound is triply exponential in $n$.
A naive implementation of the gadget would then lead to a
program with triply exponential counter values, but exponential program size in $n$.

In order to argue later that the resulting program can be encoded in a
small $\TDPN$, we will present the Lipton construction in a different
way.  Instead of growing the program with every gadget transformation,
we implement the gadgets recursively using a stack. 
We call these programs \emph{recursive net programs} ($\rnp$).
This way, when we instantiate the model for a triply exponential bound on the counters (to get $\TWOEXPSPACE$-hardness 
instead of $\EXPSPACE$-hardness), the resulting programs still have
polynomial size control flow. 
Note that at run time, such programs can have an exponentially deep stack;
however, this very large stack does not form part of the program description. 
We shall show that $\rnp$ have a natural encoding as $\TDPN$.

For the second step, we reduce termination for $\rnp$ to coverability for $\TDPN$. 
To this end, we borrow some techniques from the original construction to translate an $\rnp$ into an exponential sized Petri net. 
We then assign binary addresses to its places and construct transducers for those pairs and triples that correspond to transitions. 
This results in a $\TDPN$ of polynomial size. 
Finally, we argue that we do not need the whole exponential sized Petri net 
to reason about the transducers, and that just a polynomial size part suffices. This then gives us a polynomial time procedure.

\subsection{From Bounded Counter Programs to $\rnp$} \label{sec:MinskyToRNP}

\subparagraph{Bounded Counter Programs}
A \emph{counter program} is a finite sequence of \emph{labelled commands} separated by semicolons. 
Let $l, l_1, l_2$ be \emph{labels} and $x$ be a \emph{variable} (also called a \emph{counter}). 
The labelled commands have one of the following five forms:
\begin{alignat*}{2}
  (1)~l&: \text{\textbf{inc }}x; &&~\text{ \texttt{//} increment}\\
  (2)~l&: \text{\textbf{dec }}x; &&~\text{ \texttt{//} decrement}\\
  (3)~l&: \text{\textbf{halt}}\\
  (4)~l&: \text{\textbf{goto }} l_1; &&~\text{ \texttt{//} unconditional jump}\\
  (5)~l&: \text{\textbf{if }} x = 0 \text{\textbf{ then goto }} l_1 \text{\textbf{ else goto }} l_2; &&~\text{ \texttt{//} conditional jump}
\end{alignat*}

Variables can hold values over the natural numbers, labels have to be
pairwise distinct, but can otherwise come from some arbitrary set.
For convenience, we require each program to contain exactly one \textbf{halt} command at the very end. 
The \emph{size} $|C|$ of a counter program $C$ is the number of its labelled commands.

During execution, all variables start with initial value $0$. 
The semantics of programs follows from the syntax, except for the case of decrementing a variable whose value is already~$0$. 
In this case, the program aborts, which is different from proper termination, i.e., the execution of the \textbf{halt} command. 
It is easy to see that each counter program has only one execution, meaning it is deterministic. 
This execution is \emph{$k$-bounded} if none of the variables ever
reaches a value greater than $k$ during it. 

Let $\exp^{m+1}(x) := \exp(\exp^m(x))$ and $\exp^1(x) = \exp(x) := 2^x$.
The $N$-fold exponentially bounded \emph{halting problem} (also called \emph{termination}) for counter programs ($\HP[N]$) is given by:
\begin{description}
  \item[Input] A unary number $n \in \nats$ and a counter program $C$.
  \item[Question] Does $C$ have an $\exp^{N}(n)$-bounded execution that reaches the \textbf{halt} command?
\end{description}
We make use of the following well-known result regarding this problem:
\begin{theorem}
  For each $N > 0$, the problem $\HP[N+1]$ is $N$-$\EXPSPACE$-complete.
\end{theorem}
The proof for arbitrary $N$ matches the proof for $N = 1$, which the Lipton construction used.

\subparagraph{Recursive Net Programs}
The definition of \emph{recursive net programs} ($\rnp$) also involves sequences of labelled commands separated by semicolons. Let $l, l_1, l_2$ be labels, $x$ be a variable, and \texttt{proc} be a procedure name. Then the labelled commands can still have one of the previous forms (1) to (4). However, form (5) changes from a conditional to a nondeterministic jump, and there are two new forms for procedure calls:
\begin{alignat*}{2}
  (1)~l&: \text{\textbf{inc }}x; &&~\text{ \texttt{//} increment}\\
  (2)~l&: \text{\textbf{dec }}x; &&~\text{ \texttt{//} decrement}\\
  (3)~l&: \text{\textbf{halt}}\\
  (4)~l&: \text{\textbf{goto }} l_1; &&~\text{ \texttt{//} unconditional jump}\\
  (5)~l&: \text{\textbf{goto }} l_1 \text{\textbf{ or goto }} l_2; &&~\text{ \texttt{//} nondeterministic jump}\\
  (6)~l&: \text{\textbf{call} \texttt{proc}}; &&~\text{ \texttt{//} procedure call}\\
  (7)~l&: \text{\textbf{return}}; &&~\text{ \texttt{//} end of procedure}
\end{alignat*}

In addition to labelled commands, these programs consist of a finite set $\mathsf{PROC}$ of procedure names and also a \emph{maximum recursion depth} $k \in \mathbb{N}$. Furthermore, they not only contain one sequence of labelled commands to serve as the main program, but also include two additional sequences of labelled commands for each procedure name $\text{\texttt{proc}} \in \mathsf{PROC}$. The second sequence for each \texttt{proc} is not allowed to contain any \textbf{call} commands and serves as a sort of ``base case'' only to be called at the maximum recursion depth. Each label has to be unique among all sequences and each jump is only allowed to target labels of the sequence it belongs to. Each $\rnp$ contains exactly one \textbf{halt} command at the end of the main program. For $\text{\texttt{proc}} \in \mathsf{PROC}$ let $\#c(\texttt{proc})$ be the number of commands in both of its sequences added together and let $\#c(\texttt{main})$ be the number of commands in the main program. Then the \emph{size} of an $\rnp$ $R$ is defined as $|R| = \lceil\log{k}\rceil + \#c(\texttt{main}) + \sum_{\text{\texttt{proc}} \in \mathsf{PROC}} \#c(\texttt{proc})$.

The semantics here is quite different compared to counter programs: If
the command ``$l:\text{\textbf{call} \texttt{proc}}$'' is executed, the label $l$ gets pushed onto the call stack. Then if the stack contains less than $k$ labels, the first command sequence pertaining to \texttt{proc}, which we now call $\text{\texttt{proc}}_{<\text{max}}$, is executed. If the stack already contains $k$ labels, the second command sequence, $\text{\texttt{proc}}_{=\text{max}}$, is executed instead. Since $\text{\texttt{proc}}_{=\text{max}}$ cannot call any procedures by definition, the call stack's height (i.e.\ the \emph{recursion depth}) is bounded by $k$. On a $\text{return}$ command, the last label gets popped from the stack and we continue the execution at the label occurring right after the popped one.

How increments and decrements are executed depends on the current recursion depth $d$ as well. For each variable $x$ appearing in a command, $k + 1$ copies $x_0$ to $x_k$ are maintained during execution. The commands \textbf{inc}~$x$ resp.\ \textbf{dec}~$x$ are then interpreted as increments resp.\ decrements on $x_d$ (and not $x$ or any other copy). As before, all these copies start with value~$0$ and decrements fail at value~$0$, which is different from proper termination.

Instead of a conditional jump, we now have a nondeterministic one, that allows the program execution to continue at either label. Regarding termination we thus only require there to be at least one execution that reaches the \textbf{halt} command. This gives us the following halting problem for $\rnp$:
\begin{description}
  \item[Input] An $\rnp$ $R$
  \item[Question] Is there an execution of $R$ that reaches the \textbf{halt} command?
\end{description}

We now adapt the Lipton construction to recursive net programs. We start with a $\exp^2(n)$-bounded counter program $C$ with a set of counters $X$ and construct an $\rnp$ $R(C)$ with maximum recursion depth $n+1$ that terminates iff $C$ terminates. The number of commands in $R(C)$ will be linear in $|C|$.

\subparagraph{Auxiliary Variables}

The construction of $R(C)$ involves simulating the zero test. To this end, we introduce for each counter $x \in X$ a complementary counter $\bar{x}$ and ensure that the invariant $x_0 + \bar{x}_0 = \exp^2(n)$ always holds. We can then simulate a zero test on $x$ by checking that $\bar{x}$ can be decremented $\exp^2(n)$ times. This requires us to implement a decrement by $\exp^2(n)$ in linearly many commands and also a similar increment to reach a value of $\exp^2(n)$ for $\bar{x}$ from its initial value $0$ at the start of the program. Furthermore, we need helper variables $s$, $\bar{s}$, $y$, $\bar{y}$, $z$, and $\bar{z}$. We also sometimes need to increment or decrement the $(d+1)$th copy of one of these six variables at recursion level $d$. As an example, for incrementing $s_{d+1}$ in this way, we define the procedure \texttt{s\_inc}:
\[
  \text{\texttt{s\_inc}}_{<\text{max}}\colon~\text{\textbf{inc }}s;\text{\textbf{return}} \qquad \text{\texttt{s\_inc}}_{=\text{max}}\colon~\text{\textbf{inc }}s;\text{\textbf{return}}
\]
The analogous procedures for $\bar{s}$, $y$, $\bar{y}$, $z$, and $\bar{z}$ are defined similarly.

\subparagraph{Program Structure}

The program $R(C)$ consists of two parts: The initial part $R_{\mathit{init}}(C)$, which initializes all the complementary counters as mentioned above, followed by $R_{\mathit{sim}}(C)$, the part that simulates $C$. We construct $R_{\mathit{sim}}(C)$ from $C$ by replacing some of its commands. Increments of the form \textbf{inc}~$x$ are replaced by \textbf{dec}~$\bar{x};$\textbf{inc}~$x$, decrements \textbf{dec}~$x$ are replaced by \textbf{dec}~$x;$\textbf{inc}~$\bar{x}$. Unconditional jumps and the \textbf{halt} command stay the same. Each conditional jump (form (5) for counter programs)
is replaced by
\begin{align*}
  l&: \text{Test}(x, l_\text{continue}, l_2);\\
  l_\text{continue}&: \text{Test}(\bar{x}, l_1, l_2)
\end{align*}
where Test$(x,l_{zero},l_{nonzero})$ is what we call a \emph{macro}. We use it as syntactic sugar to be replaced by its specification for the actual construction of $R(C)$. This is in contrast to procedures, which refer to specific parts of the program that can be called to increase the recursion depth.

\subparagraph{Test Macros and Decrement Procedure}

The macro Test is specified in the left part of Figure~\ref{fig:Testdec}. It involves a call to the procedure \texttt{dec}, which is defined in the right part of the same figure. Below Test we have also specified the variant Test$_{+1}$, which is used in \texttt{dec}. The main difference is that Test$_{+1}$ can only be invoked on variables $y$ or $z$ and acts on their $(d+1)$th copy at recursion depth $d$.

\begin{figure}[t]
  \centering
  {\small
  \begin{minipage}[t]{.4\textwidth}
    \centering
    \begin{align*}
      \text{Test}&(x, l_\text{zero}, l_\text{nonzero}):\\[1em]
      &~\text{\textbf{goto }} l_\text{nztest} \text{\textbf{ or goto }} l_\text{loop};\\
      l_\text{nztest}:&~ \text{\textbf{dec }}x;\text{\textbf{inc }}x;\text{\textbf{goto }} l_\text{nonzero};\\
      l_\text{loop}:&~ \text{\textbf{dec }}\bar{x};\text{\textbf{inc }}x;\text{\textbf{call} }\bar{\text{\texttt{s}}}\text{\texttt{\_dec}};\text{\textbf{call} \texttt{s\_inc}};\\
      &~\text{\textbf{goto }} l_\text{exit} \text{\textbf{ or goto }} l_\text{loop};\\
      l_\text{exit}:&~ \text{\textbf{call} \texttt{dec}};\text{\textbf{goto }} l_\text{zero}\\[2em]
      \text{Test}_{+1}&(v, l_\text{zero}, l_\text{nonzero}):\\[1em]
      &~\text{\textbf{goto }} l_\text{nztest} \text{\textbf{ or goto }} l_\text{loop};\\
      l_\text{nztest}:&~ \text{\textbf{call} \texttt{v\_dec}};\text{\textbf{call} \texttt{v\_inc}};\text{\textbf{goto }} l_\text{nonzero};\\
      l_\text{loop}:&~ \text{\textbf{call} }\bar{\text{\texttt{v}}}\text{\texttt{\_dec}};\text{\textbf{call} \texttt{v\_inc}};\\
      &~ \text{\textbf{call} }\bar{\text{\texttt{s}}}\text{\texttt{\_dec}};\text{\textbf{call} \texttt{s\_inc}};\\
      &~\text{\textbf{goto }} l_\text{exit} \text{\textbf{ or goto }} l_\text{loop};\\
      l_\text{exit}:&~ \text{\textbf{call} \texttt{dec}};\text{\textbf{goto }} l_\text{zero}
    \end{align*}
  \end{minipage}
  \begin{minipage}[t]{.4\textwidth}
    \centering
    \begin{align*}
      \text{\texttt{dec}}_{<\text{max}}&:\\[1em]
      l_\text{outer}:&~ \text{\textbf{call} \texttt{y\_dec}};\text{\textbf{call} }\bar{\text{\texttt{y}}}\text{\texttt{\_inc}};\\
      l_\text{inner}:&~ \text{\textbf{call} \texttt{z\_dec}};\text{\textbf{call} }\bar{\text{\texttt{z}}}\text{\texttt{\_inc}};\\
      &~ \text{\textbf{dec }}s;\text{\textbf{inc }}\bar{s};\\
      &~\text{Test}_{+1}(z, l_\text{next}, l_\text{inner});\\
      l_\text{next}:&~ \text{Test}_{+1}(y, l_\text{exit}, l_\text{outer});\\
      l_\text{exit}:&~ \text{\textbf{return}}\\[2em]
      \text{\texttt{dec}}_{=\text{max}}&:\\[1em]
      &~ \text{\textbf{dec }}s;\text{\textbf{inc }}\bar{s};\text{\textbf{dec }}s;\text{\textbf{inc }}\bar{s};\\
      &~ \text{\textbf{return}}
    \end{align*}
  \end{minipage}
  }
  \caption{Definitions of the macros Test and Test$_{+1}$ as well as the procedure \texttt{dec}.
  Regarding the second macro we require $v \in \{y,z\}$.}
  \label{fig:Testdec}
\end{figure}

Semantically, \texttt{dec} at recursion depth $d$ decrements $s_d$ by $\exp^2(n+1-d)$ (and increments $\bar{s}_d$ by the same amount). Both variants of Test simulate a conditional jump and have the side effect of switching the values $x_d$ and $\bar{x}_d$ if the tested variable $x_d$ was $0$. Because of this, every conditional jump of $C$ gets replaced by two instances of the Test macro, where the second one reverses the potential side effect. 

The decrements of procedure \texttt{dec} are performed via two nested loops that each run $\exp^2(n-d)$-times. Each of these loops uses a helper variable $y_{d+1}$ or $z_{d+1}$ that has to be tested for zero at the end, using the Test$_{+1}$ macro. This involves transferring the helper variable's value to $s_{d+1}$ and then calling \texttt{dec} at the next recursion depth. Essentially, any decrement by $\exp^2(j)$ for some $j$ is implemented using $\exp^2(j-1)$ many decrements by $\exp^2(j-1)$ via the nested loops. This iterative squaring of the value by which we decrement continues down to the base case of $\exp^2(0) = 2$.

\subparagraph{Semantics}

Our construction is semantically very similar to the Lipton construction, barring two main differences: Firstly, instead of having $n+1$ different procedure definitions of \texttt{dec} (one per level $d$), we only need two because of recursion. The case for the Test macros is similar, as is the case of the helper variables $s$, $y$, $z$ and their complements. Secondly, our variable copies start with index $0$ counting upwards, whereas in the Lipton construction the variables start with index $n$ and count downwards. This means that for some index $d$ we have the invariant $s_d + \bar{s}_d = \exp^2(n+1-d)$ in our construction, where it is $s_d + \bar{s}_d = \exp^2(d)$ for Lipton. While the invariant of the Lipton construction is simpler, ours allows us to define the recursion depth starting at 0 and going upwards, which seemed more natural for recursion.

Let us give a more precise analysis regarding the effect of the Test macros and \texttt{dec} procedure. During the execution of \texttt{dec} at recursion depth $d$, we begin with $s_d = \exp^2(n+1-d)$, $y_{d+1} = z_{d+1} = \exp^2(n-d)$, and $\bar{s}_d = \bar{y}_{d+1} = \bar{z}_{d+1} = 0$. The invariants $s_d + \bar{s}_d = \exp^2(n+1-d)$, $y_{d+1} + \bar{y}_{d+1} = \exp^2(n-d)$, and $z_{d+1} + \bar{z}_{d+1} = \exp^2(n-d)$ are upheld throughout. At the end we have $s_d = \bar{y}_{d+1} = \bar{z}_{d+1} = 0$, $y_{d+1} = z_{d+1} = \exp^2(n-d)$, and $\bar{s}_d = \exp^2(n+1-d)$, meaning the decrements were performed correctly and all helper variables retain their initial values. The situation is quite similar for Test and Test$_{+1}$, if the variable to be tested was initially $0$. In the non-zero case, the tested variable is just decremented and incremented once, whereas no other variables are touched. All executions that differ from the described behavior are guaranteed to get stuck.

Correctness of these semantics is proven by induction on the recursion depth in Appendix~\ref{sec:decinduction}.
It requires the assumptions $x_0 + \bar{x}_0 = \exp^2(n)$, $v_{d} + \bar{v}_{d} = \exp^2(n+1-d)$, and $\bar{v}_{d} = 0$ for all $x \in X$, $v \in \{\bar{s},y,z\}$, and $d > 0$.

\subparagraph{Initialization}

We now have to construct $R_{\mathit{init}}(C)$ in such a way, that it
performs all the necessary increments for these assumptions to hold at
the start of $R_{\mathit{sim}}(C)$. It has the following form:

\begin{alignat*}{2}
  &~ \text{\textbf{call} \texttt{inc}}; &&\\
  l_\text{loop}:&~ \text{\textbf{call} \texttt{y\_dec}};\text{\textbf{call} }\bar{\text{\texttt{y}}}\text{\texttt{\_inc}}; &&\\
  &~ \text{\textbf{inc }}\bar{x};\ldots &&\text{ \texttt{//} copy for each }x \in X\\
  &~ \text{Test}_{+1}(y, l_\text{exit}, l_\text{loop}); &&\\
  l_\text{exit}:&~ \ldots &&\text{ \texttt{//} first command of }R_{\mathit{sim}}(C)
\end{alignat*}
Here, \texttt{inc} is the procedure defined in Figure \ref{fig:inc}. Semantically, it performs the correct amount of increments for all copies of $\bar{s}$, $\bar{y}$ and $\bar{z}$. This is again achieved using iterative squaring. Because \texttt{inc}'s first command is a call to itself, it first handles all the variable copies at higher recursion depth, before continuing at the current recursion depth $d$. Therefore using $\bar{s}_{d+1}$, $\bar{y}_{d+1}$ and $\bar{z}_{d+1}$ as part of the Test$_{+1}$ macro does not cause any problems.

Following the call to \texttt{inc}, the remainder of $R_{\mathit{init}}(C)$ then performs the correct increments regarding $\bar{x}_0$ for each $x \in X$. Since $y_1$ had the same target value as these, we use it as a counter for the loop that realizes these last increments. The Test$_{+1}$ macro at the end then conveniently resets the value of $y_1$, once it is decremented to $0$.

Like for \texttt{dec} and the Test macros, the proof of correctness regarding these semantics requires an induction on the recursion depth and can be found in Appendix~\ref{sec:decinduction}.

\begin{figure}[t]
  {\small
  \centering
  \begin{minipage}{.4\textwidth}
    \centering
    \begin{align*}
      \text{\texttt{inc}}_{<\text{max}}&:\\[1em]
      &~ \text{\textbf{call} \texttt{inc}};\\
      l_\text{outer}:&~ \text{\textbf{call} \texttt{y\_dec}};\text{\textbf{call} }\bar{\text{\texttt{y}}}\text{\texttt{\_inc}};\\
      l_\text{inner}:&~ \text{\textbf{call} \texttt{z\_dec}};\text{\textbf{call} }\bar{\text{\texttt{z}}}\text{\texttt{\_inc}};\\
      &~ \text{\textbf{inc }}y;\text{\textbf{inc }}z;\text{\textbf{inc }}\bar{s};\\
      &~ \text{Test}'_{+1}(z, l_\text{next}, l_\text{inner});\\
      l_\text{next}:&~ \text{Test}'_{+1}(y, l_\text{exit}, l_\text{outer});\\
      l_\text{exit}:&~ \text{\textbf{return}}
    \end{align*}
  \end{minipage}
  \hspace{1em}
  \begin{minipage}{.4\textwidth}
    \centering
    \begin{align*}
      \text{\texttt{inc}}_{=\text{max}}&:\\[1em]
      &~ \text{\textbf{inc }}y;\text{\textbf{inc }}y;\\
      &~ \text{\textbf{inc }}z;\text{\textbf{inc }}z;\\
      &~ \text{\textbf{inc }}\bar{s};\text{\textbf{inc }}\bar{s};\\
      &~ \text{\textbf{return}}
    \end{align*}
  \end{minipage}
  }
  \caption{Definition of the procedure \texttt{inc}.}
  \label{fig:inc}
\end{figure}

\subparagraph{Size Analysis}
To give a brief size analysis of $R(C)$, $\mathsf{PROC}$ contains 14 procedure names, whose corresponding definitions have constant size. For each command in $C$, $R_{\mathit{sim}}(C)$ contains constantly many commands, and $R_{\mathit{init}}(C)$ has linearly many commands in the size of the variable set $X$. Since wlog.\ each variable of $C$ is involved in at least one of its commands, the amount of commands in $R(C)$ is linear in $|C|$. Here, for doubly exponential counter values, we would not even need $n$ to be given in unary since only $\lceil\log(n+1)\rceil$ factors into the size of $R(C)$.

\subparagraph{Handling Triply Exponential Counter Values}
The exact same construction with a maximum recursion depth of $2^n + 1$ can be used to simulate a counter program with counters bounded by $\exp^3(n)$: Starting with $2$ and squaring $n$-times yields $\exp^2(n)$, therefore squaring $2^n$ times instead yields $\exp^3(n)$. The correctness follows from the same inductive proofs as before. For this changed maximum recursion depth, configurations contain exponentially in $n$ many counter values and also maintain a call stack of size up to $2^n$. However, since the maximum recursion depth can be encoded in binary, its size is still polynomial in the unary encoding of $n$. Thus, the halting problem for recursive net programs is $\TWOEXPSPACE$-hard.

\subsection{From $\rnp$ to $\TDPN$} \label{sec:RNPtoTPN}

Figure~\ref{fig:netcommands} and Figure~\ref{fig:netprocedures} show how the commands of recursive net programs can be simulated by Petri net transitions. This is again done in similar fashion to Esparza's description~\cite{Esp98a} of the Lipton construction~\cite{lipton1976reachability}. As we can see, this involves only the three types of transitions defined by our transducers.

\begin{figure}[t]
  \centering
  \def \dist{1}
{\small
  \begin{tikzpicture}[scale=0.8]
    \node[place,label=right:$l_{1,d}$] (p1) at (0,0) {};
    \node[place,label=right:$l_{2,d}$] (p2) at (0,-2*\dist) {};
    \node[place,label=right:$x_{d}$] (p3) at (1.5,-1*\dist) {};
    \node[transition] at (0,-1*\dist) {}
      edge[pre] (p1)
      edge[post] (p2)
      edge[post] (p3);
    \node at (0.5, -2*\dist - 1.2) {$\begin{aligned} l_1&: \text{\textbf{inc }}x; \\ l_2&: \ldots \end{aligned}$};
    
    \node[place,label=right:$l_{1,d}$] (q1) at (3.5,0) {};
    \node[place,label=right:$l_{2,d}$] (q2) at (3.5,-2*\dist) {};
    \node[place,label=right:$x_{d}$] (q3) at (5,-1*\dist) {};
    \node[transition] at (3.5,-1*\dist) {}
      edge[pre] (q1)
      edge[post] (q2)
      edge[pre] (q3);
    \node at (4, -2*\dist - 1.2) {$\begin{aligned} l_1&: \text{\textbf{dec }}x; \\ l_2&: \ldots \end{aligned}$};
    
    \node[place,label=right:$l_{1,d}$] (t1) at (7,0) {};
    \node[place,label=right:$w_{halt}$] (t2) at (7,-2*\dist) {};
    \node[transition] at (7,-1*\dist) {}
      edge[pre] (t1)
      edge[post] (t2);
    \node at (7.25, -2*\dist - 1.2) {$l_1\colon \text{\textbf{halt}};$};

    \node[place,label=right:$l_{1,d}$] (r1) at (10, 0) {};
    \node[place,label=right:$l_{2,d}$] (r2) at (10,-2*\dist) {};
    \node[transition] at (10,-1*\dist) {}
      edge[pre] (r1)
      edge[post] (r2);
    \node at (10.25, -2*\dist - 1.2) {$l_1\colon \text{\textbf{goto }} l_2;$};
    
    \node[place,label=right:$l_{1,d}$] (s1) at (14,0) {};
    \node[place,label=right:$l_{2,d}$] (s2) at (13,-2*\dist) {};
    \node[place,label=right:$l_{3,d}$] (s3) at (15,-2*\dist) {};
    \node[transition] at (13,-1*\dist) {}
      edge[pre] (s1)
      edge[post] (s2);
    \node[transition] at (15,-1*\dist) {}
      edge[pre] (s1)
      edge[post] (s3);
    \node at (14.25, -2*\dist - 1.2) {$l_1\colon \text{\textbf{goto }} l_2 \text{\textbf{ or goto }} l_3;$};
  \end{tikzpicture}
}
  \caption{Petri net transitions for five of the seven command types found in recursive net programs. Here, $d \in \{0,\ldots,k\}$, where $k$ is the maximum recursion depth.}
  \label{fig:netcommands}
\end{figure}
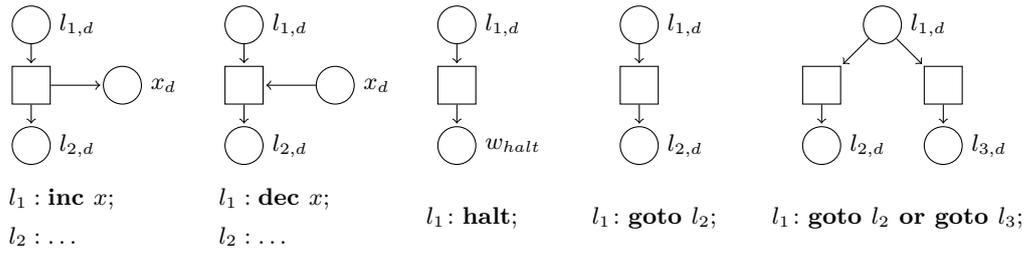

\begin{figure}[t]
  \def \dist{1.5}
  \centering
  \begin{tikzpicture}[]
    \node[place,label=above:$l_{1,d}$] (l1) at (-2,0) {};
    \node[place,label=above:$l_{1,d}$\_calls\_\texttt{proc}] (l1callsl3) at (2,0) {};
    \node[place,label=above:$l_{2,d}$] (l2) at (6,0) {};
    
    \node[place,label=below:$l_{3,d+1}$] (l3) at (0,-\dist) {};
    \node at (1,-\dist) {$\cdots$};
    \node[place,label=below:$l_{4,d+1}$] (l4) at (2,-\dist) {};
    \node[place,label=below right:return\_\texttt{proc}$_{d+1}$] (returnl3) at (4,-\dist) {};
    \node[transition] at (3,-\dist) {}
      edge[pre] (l4)
      edge[post] (returnl3);
    \node[right] at (7, -\dist) {$\begin{aligned} \text{\texttt{proc}}\colon~l_3:&~\ldots \\ &~\vdots \\ l_4:&~\text{\textbf{return}}; \end{aligned}$};
    
    \node[transition] at (-1,0) {}
      edge[pre] (l1)
      edge[post] (l3)
      edge[post] (l1callsl3);
    \node[transition] at (5,0) {}
      edge[pre] (l1callsl3)
      edge[pre] (returnl3)
      edge[post] (l2);
    \node[right] at (7,0) {$\begin{aligned}
      l_1&: \text{\textbf{call} \texttt{proc}}; \\ l_2&: \ldots \end{aligned}$};
  \end{tikzpicture}

  \caption{Petri net transitions for procedure calls found in recursive net programs. Here, $d \in \{0,\ldots,k-1\}$, where $k$ is the maximum recursion depth.}
  \label{fig:netprocedures}
\end{figure}
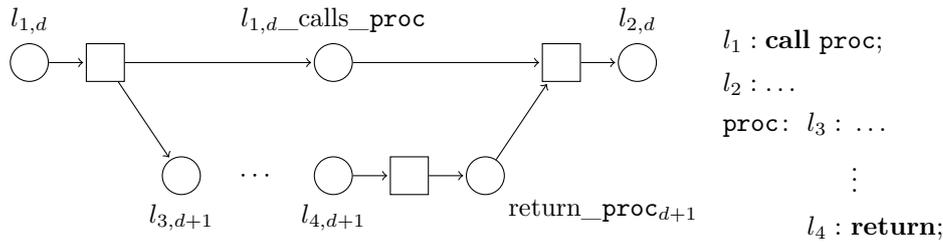

Let us give more detail regarding the Petri net construction: Given an
$\rnp$ $R$ with maximum recursion depth $k$ we construct a
transducer-defined Petri net $\mathcal{N} = (w_{\mathit{init}},$ $w_{\mathit{final}},$ $\mathcal{T}_{\mathit{move}},$ $\mathcal{T}_{\mathit{fork}},$ $\mathcal{T}_{\mathit{join}})$, which defines the Petri net $N(\cN) = (P,T,F,p_0,p_f)$, such that $\multi{p_f}$ is coverable in $N(\mathcal{N})$ iff there is a terminating execution of $R$. We begin by arguing about the shape of $N(\mathcal{N})$ and then construct our transducers afterwards.

The idea is for $N(\mathcal{N})$ to start with one place per variable and one place per label, as well as one auxiliary place for each \textbf{call} command and each $\text{\texttt{proc}} \in \mathsf{PROC}$, which can be seen in Figure~\ref{fig:netprocedures}. Additionally, there is also a single auxiliary place $w_{\mathit{halt}}$ for the \textbf{halt} command. Let the number of all these places be $h$. Then each such place gets copied $k+1$ times, so that a copy exists for each possible recursion depth. Transitions get added at each recursion depth $d$ according to Figure~\ref{fig:netcommands} and Figure~\ref{fig:netprocedures}, whereas some transitions in the latter also connect to places of recursion depth $d+1$.

Regarding the transducers, we use the alphabet $\{0,1\}$. Every place address $w = u.v$ has a prefix $u$ of length $\lceil\log{h}\rceil$ and a postfix $v$ of length $\lceil\log{k}\rceil$. We assign each of the $h$ places that $N(\mathcal{N})$ started with a number from $0$ to $h-1$. The binary representation of this number (with leading zeros) is used for the $u$-part of its address. For the $v$-part, we use the binary representation of the recursion depth $d$ (also with leading zeros), that a particular copy of this place corresponds to. The address of the place corresponding to the first label in the main program at recursion depth $0$ is used for $w_{\mathit{init}}$, whereas the one corresponding to $w_{\mathit{halt}}$ at recursion depth $0$ is used for $w_{\mathit{final}}$.

To accept a particular pair or triple of addresses as a transition, each of the three transducers distinguishes between all possibilities regarding the $u$-parts. Any pair or triple of $\lceil\log{h}\rceil$-length words that matches a particular transition of the right type (move, fork, join) has a unique path in the transducer, while all non-matching pairs or triples do not. Then for the $v$-parts, the transducer needs to either check for equality, if all places correspond to the same recursion depth, or for one binary represented number to be one higher. Since it is clear from the $u$-parts, whether the recursion depths should all match or not, we can just connect the unique paths to the correct part of the transducer at the end.

The transducer parts for the first $\lceil\log{h}\rceil$ bits require to distinguish between up to\linebreak $2^{3\log{h}} = 8h$ possibilities, meaning they require polynomially in $h$ many states. The parts for the last $\lceil\log{k}\rceil$ bits can easily be constructed using polynomially many states in $\log{k}$.
More details on this construction can be found in Appendix~\ref{sec:TPNconstruction}.
Since $h$ is linear in the number of commands in $R$ and $\lceil\log{k}\rceil$ is the size of the binary encoding of the maximum recursion depth, $\mathcal{N}$ is of polynomial size compared to $R$. Because we can construct $\mathcal{N}$ by first constructing $N(\mathcal{N})$ without the copies for each recursion depth, this is feasible in polynomial time. Thus, the coverability problem for transducer-defined Petri nets is $\TWOEXPSPACE$-hard.

\section{Discussion}
\label{sec:discussion}

The chain of reductions in Sections~\ref{sec:tdpn_to_dcps} and~\ref{sec:lower_bound}
complete the $\TWOEXPSPACE$ lower bound for $1$-bounded reachability for $\DCPS$.
In fact, an inspection of the reductions show a technical strengthening: 
the $\TWOEXPSPACE$ lower bound already holds for $\SRP[1]$ of $\DCPS$ 
which satisfy two additional properties, \emph{boundedness} and \emph{local termination}. 

\begin{definition}
	A $\DCPS$ $\cA$ is said to be \emph{bounded} if there is a global
	bound $B \in \nats$ on the size of every configuration of every run of $\cA$.
	It is \emph{locally terminating} if every
	infinite run of $\cA$ contains infinitely many context switches.
      \end{definition}

Consider the chain of reductions from the halting problem for bounded counter programs
to $\rnp$ to $\TDPN$ to $\SRP[1]$.
The configurations of the counter programs, by definition, are bounded by a triply-exponential
bound on the parameter $n$.
This bound translates to bounds on the $\rnp$ and $\TDPN$ instances.
In particular, the number of places in the $\TDPN$ produced in the reduction
is exponentially bounded in $n$ and the number of
tokens on these places is triple-exponentially bounded in $n$. 
The $\DCPS$ constructed from the $\TDPN$ uses the stack of a thread to store
an address of a place; thus, the height of a stack is bounded by a polynomial in $n$.
In addition, since the number of tokens in the $\TDPN$ correspond to the number of in-progress
threads in the $\DCPS$, this implies a triple exponential (in $n$) bound on the number of threads
in any execution of the $\DCPS$. 
Thus, the size of every configuration in every run of the $\DCPS$ is bounded.

Second, the rules of the constructed $\DCPS$ do not allow any
one thread to run indefinitely. In other words, any non-terminating
run of the $\DCPS$ must involve infinitely many threads and the run
contains infinitely many context switches.

\begin{theorem}
	The $\SRP[1]$ problem for bounded, locally terminating $\DCPS$ is $\TWOEXPSPACE$-hard.
\end{theorem}

\bibliography{bibliography}

\appendix

\section{Proofs for Section \ref{sec:tdpn_to_dcps}}
\label{appendix:killDCPS}
For simplicity, we assume that threads die when they have an empty
stack.

In order to explain
our simulation of a $\TDPN$ by a $\DCPS$, we introduce the following
extension of $\DCPS$ which is easily seen to be syntactic sugar in
that it
does not add any more power to the $\DCPS$ model.
\begin{definition}
\label{def:kill_dcps}
	A \textit{kill}-$\DCPS$ $\cA=(G,\Gamma,\Delta,g_0,\gamma_0)$ is a
	$\DCPS$ with following additional features:

	\begin{itemize}
	 	\item $\Gamma=\Gamma_{\mathit{reg}}
	\cup \Gamma_{\mathit{kill}}$ is the
	disjoint union of two sets of symbols.
	\item Only finite state transitions are allowed on the set $\Gamma_{\mathit{kill}}$. In other words, for any $\gamma \in \Gamma_{\mathit{kill}}$, for any
	rule of the form 
	\[ g| \gamma \hookrightarrow g'|w \triangleright
	\; \gamma'' \quad \quad 
	\text{ or } \quad \quad
	 g| \gamma \hookrightarrow g'|w\]
	it is the case that $w \in (\Gamma_{\mathit{kill}} \cup \varepsilon)$.
	\item We allow {kill}-rules of
	the
	following form in $\Delta$ for $\gamma,\gamma' \in \Gamma_{\mathit{kill}}$:
	\[ g| \gamma \hookrightarrow g'|\gamma \not \triangleright
	\; \gamma' \quad \quad \text{or} \quad \quad  g| \gamma
	\hookrightarrow g'|\varepsilon \not \triangleright
	\; \gamma'\]
Corresponding to such rules,	the relation $\Rightarrow_i$ for $\cA$
 additionally includes the
	following transitions:
	\[ \langle g,(\gamma,i), \mmap \oplus [(\gamma',j)] \rangle
	\rightarrow_i \langle
	g',(\gamma,0),\mmap \rangle \quad \text{or} \quad \langle g,(\gamma,i), \mmap \oplus [(\gamma',j)] \rangle
	\rightarrow_i \langle
	g',(\varepsilon,0),\mmap \rangle\]

	 \end{itemize} 
	 Note that the number of context switches of the active thread
	 drops
	 to zero on application of a kill rule.
	 
	 The relation $\Rightarrow_{\leq k}$ and its transitive closure
	 $\Rightarrow_{\leq k}^*$ are defined as for a $\DCPS$, with the
	 additional restriction that the killed thread $
	 (\gamma',j)$ in every application of a kill rule must satisfy $j
	 \leq
	 k$. A
	 configuration $c$ of $\cA$ is said to be $k$-bounded reachable by a
	 kill-$\DCPS$ if $\langle g_0,(\gamma_0,0),\emptyset \rangle
	 \Rightarrow_{\leq k}^* c$ with this restriction.
\end{definition}

\begin{proposition}
\label{prop:kill_simulation}
Reachability for kill-$\DCPS$ is reducible in polynomial time to
reachability for $\DCPS$.
\end{proposition}
We show that for every kill-$\DCPS$ $\cA$, there exists a $\DCPS$
$\cA'$ such that
	$\cA'$ can be produced in time polynomial in the size of $\cA$ and if
	$\mathrm{config}(\cA),\mathrm{config}
	(\cA')$ are the
	set of configurations of $\cA,\cA'$ respectively, then
	$\mathrm{config}(\cA)
	\subseteq \mathrm{config}(\cA')$ and
	any configuration $c \in \mathrm{config}(\cA)$ is $k$-bounded reachable in
	$\cA$ iff $c$
	is $k$-bounded
	reachable in
	$\cA'$.

\begin{proof}
	Let $\cA=(G,\Gamma,\Delta,g_0,\gamma_0)$ where $\Gamma=\Gamma_{\mathit{reg}} \cup
	\Gamma_{\mathit{kill}}$, then there can be at most
	$O(|Q|^2|\Gamma^2|)$ many kill rules.  Let $\cA'=
	(G',\Gamma \cup \{ \gamma_{\mathit{spawn}}\},\Delta',g_0,\gamma_0)$ where
	$G'=G \cup G_{\mathit{kill}}$ with
	$G_{\mathit{kill}}=(\{g_{\mathit{kill}},g_{\mathit{return}},g_{\mathit{spawn}}\}\times Q^2 \times \Gamma^2 \times \{nopop,pop\})
	$, $\Delta'=\Delta \cup \Delta_{\mathit{kill}}$ where $\Delta_{\mathit{kill}}$ is
	described below.

We demonstrate for the case
	corresponding to a kill-rule $g| \gamma \hookrightarrow g'|\gamma
	\not \triangleright
	\; \gamma'$. The case where $\gamma$ is popped and results in
	$\varepsilon$ uses the states with the marker $pop$ in the global
	state with appropriate
	modifications. We include the
	following set of transitions
	in
	$\Delta_{\mathit{kill}}$:
	\begin{enumerate}
		\item $g| \gamma \hookrightarrow (g_{\mathit{spawn}},g,g',\gamma,\gamma',nopop)|\gamma \triangleright
			\gamma_{\mathit{spawn}}$ 
			\item $(g_{\mathit{spawn}},g,g',\gamma,\gamma',nopop)| \gamma \hookrightarrow (g_{\mathit{kill}},g,g',\gamma,\gamma',nopop)|\varepsilon$
			\item $(g_{\mathit{kill}},g,g',\gamma,\gamma',nopop) | \gamma'
			\hookrightarrow (g_{\mathit{return}},g,g',\gamma,\gamma',nopop)|\varepsilon$
			\item $(g_{\mathit{return}},g,g',\gamma,\gamma',nopop) | \gamma_{\mathit{spawn}}
			\hookrightarrow
			g'|\gamma$.
	\end{enumerate}
	Note that in the case of $\gamma$ being popped, Step 4 above is
	modified so that the right hand side is $g'|\varepsilon$. 
Since we add only finitely many new rules to $\cA'$ for each kill-rule
of $\cA$, $\Delta'$ is bigger than $\Delta$ by $O
(|Q|^2|\Gamma|^2)$, $Q'$ is bigger than $Q$ by $O
(|Q|^2|\Gamma|^2)$, the new alphabet size is one more than the old
alphabet size and hence
$\cA'$ is polynomial in the size of $\cA$.

It is clear that the application of any kill-rule in $\cA$ can be
simulated by
the extra rules that we have added to $\cA'$. Corresponding to any
transition in $\cA$ of the
form
\[ \langle g,(\gamma,i), \mmap \oplus \multi{(\gamma',j)} \rangle
	\overset{\cA}\rightarrow_i \langle
	g',(\gamma,0),\mmap \rangle,\]
we have the following sequence of transitions in $\cA'$
\begin{flalign*}
	&\langle g,(\gamma,i), \mmap \oplus \multi{(\gamma',j)} \rangle\\
\overset{\cA'}\rightarrow_i &\langle (g_{\mathit{spawn}},g,g',\gamma,\gamma',nopop),(\gamma,i), \mmap \oplus \multi{
(\gamma_{\mathit{spawn}},0),(\gamma',j)}\rangle\\
\overset{\cA'}\rightarrow_i &\langle (g_{\mathit{kill}},g,g',\gamma,\gamma',nopop),(\varepsilon,i), \mmap \oplus \multi{
(\gamma_{\mathit{spawn}},0),(\gamma',j)}\rangle\\
\overset{\cA'}\mapsto_i &\langle (g_{\mathit{kill}},g,g',\gamma,\gamma',nopop),(\gamma',j), \mmap \oplus \multi{
(\gamma_{\mathit{spawn}},0)} \rangle\\
\overset{\cA'}\rightarrow_j &\langle (g_{\mathit{return}},g,g',\gamma,\gamma',nopop),(\varepsilon,j), \mmap \oplus \multi{
(\gamma_{\mathit{spawn}},0)}\rangle\\
\overset{\cA'}\mapsto_j &\langle (g_{\mathit{return}},g,g',\gamma,\gamma',nopop),(\gamma_{\mathit{spawn}},0), \mmap  \rangle\\
\overset{\cA'}\rightarrow_0 &\langle g',(\gamma,0), \mmap \rangle.
\end{flalign*}
Conversely, we will prove by induction on the number of configurations
of the form $\langle (g_{\mathit{spawn}},$ $g,$ $g',$ $\gamma,$ $\gamma'),$ $(\gamma,$ $i),$ $\mmap
\rangle$ (i.e. those which contain $g_{\mathit{spawn}}$ in their state) in a
run $\rho'$ of $\cA'$ that
$k$-bounded reaches a
configuration
$\langle g,(\gamma,i),\mmap \rangle \in \mathrm{config}(\cA)$ that
there is a
corresponding $k$-bounded run $\rho$ of $\cA$ which reaches $\langle
q,(\gamma,i),\mmap \rangle$.\\
In the base case, there are no configurations involving $g_{\mathit{spawn}}$.
This means that only the rules in $\Delta$ are used in $\rho'$ and the
required witness $\rho=\rho'$.\\
Suppose we have $i+1$ occurences of $g_{\mathit{spawn}}$ in $\rho'$. Let 
\begin{flalign}
 	\rho'= & c_0
\overset{\cA'}\Rightarrow_{\leq k}^* \langle g,(\gamma,i'),\mmap \oplus
\multi{(\gamma',j'),(\gamma'',j'')}
\rangle\\
\overset{\cA'}\rightarrow_{i'} &\langle (g_{\mathit{spawn}},g,g',\gamma,\gamma',nopop),(\gamma,i'), \mmap \oplus \multi{
(\gamma,0),(\gamma',j'),(\gamma'',j'')}\rangle\\
\overset{\cA'}\rightarrow_{i'} &\langle (g_{\mathit{kill}},g,g',\gamma,\gamma',nopop),(\varepsilon,i'), \mmap \oplus \multi{
(\gamma,0),(\gamma',j'),(\gamma'',j'')}\rangle\\
\overset{\cA'}\mapsto_{i'} &\langle (g_{\mathit{kill}},g,g',\gamma,\gamma',nopop),(\gamma',j'), \mmap \oplus \multi{
(\gamma,0),(\gamma'',j'')}\rangle\\
\overset{\cA'}\mapsto_{j'} &\langle (g_{\mathit{kill}},g,g',\gamma,\gamma',nopop),(\gamma'',j''), \mmap \oplus \multi{
(\gamma,0),(\gamma',j'+1)} \rangle\\
\overset{\cA'}\mapsto_{j''} &\langle (g_{\mathit{kill}},g,g',\gamma,\gamma',nopop),
(\gamma',j'+1), \mmap \oplus \multi{
(\gamma,0),(\gamma'',j''+1)} \rangle\\
\overset{\cA'}\rightarrow_{j'+1} &\langle (g_{\mathit{return}},g,g',\gamma,\gamma',nopop),(\varepsilon,j'+1), \mmap \oplus \multi{
(\gamma,0),(\gamma'',j''+1)}\rangle\\
\overset{\cA'}\mapsto_{j'+1} &\langle (g_{\mathit{return}},g,g',\gamma,\gamma' ,nopop),(\gamma,0), \mmap \oplus \multi{
		(\gamma'',j''+1)} \rangle\\
\overset{\cA'}\rightarrow_0 &\langle g',(\gamma,0), \mmap \oplus \multi{
(\gamma'',j''+1)} \rangle
 \end{flalign} 
 Note that the only transitions that could happen
 between lines 6 and 9 are context switches by construction i.e. none
 of the contents of any of the threads can be affected in any way.
 This fact also holds for context switches at the $g_{\mathit{spawn}}$ and
 $g_{\mathit{return}}$ states. The
 context switches in lines 7 and 8 above could be removed to get the
 following run:
\begin{flalign*}
 	\rho''= & c_0
\overset{\cA'}\Rightarrow_{\leq k}^* \langle g,(\gamma,i'),\mmap \oplus
\multi{(\gamma',j'),(\gamma'',j'')}
\rangle\\
\overset{\cA'}\rightarrow_{i'} &\langle (g_{\mathit{spawn}},g,g',\gamma,\gamma',nopop),(\gamma,i'), \mmap \oplus \multi{
(\gamma,0),(\gamma',j'),(\gamma'',j'')}\rangle\\
\overset{\cA'}\rightarrow_{i'} &\langle (g_{\mathit{kill}},g,g',\gamma,\gamma',nopop),(\varepsilon,i'), \mmap \oplus \multi{
(\gamma,0),(\gamma',j'),(\gamma'',j'')}\rangle\\
\overset{\cA'}\mapsto_{i'} &\langle (g_{\mathit{kill}},g,g',\gamma,\gamma',nopop),(\gamma',j'), \mmap \oplus \multi{
(\gamma,0),(\gamma'',j'')}\rangle\\
\overset{\cA'}\rightarrow_{j'} &\langle (g_{\mathit{return}},g,g',\gamma,\gamma',nopop),(\varepsilon,j'), \mmap \oplus \multi{
(\gamma,0),(\gamma'',j'')}\rangle\\
\overset{\cA'}\mapsto_{j'} &\langle (g_{\mathit{return}},g,g',\gamma,\gamma',nopop),(\gamma,0), \mmap \oplus \multi{(\gamma'',j'')} \rangle\\
\overset{\cA'}\rightarrow_0 &\langle g',(\gamma,0), \mmap \oplus \multi{
(\gamma'',j'')} \rangle
 \end{flalign*}
 By induction hypothesis, there exists a run $\tilde{\rho}$ in $\cA$
 such that $c_0
\overset{\cA}\Rightarrow_{\leq k}^* \langle g,(\gamma,i'),\mmap \oplus
\multi{(\gamma',j'),(\gamma'',j'')}
\rangle$ and we also have $\langle g,(\gamma,i'),\mmap \oplus
\multi{(\gamma',j'),(\gamma'',j'')}
\rangle \overset{\cA}\rightarrow_{j'} \langle g',(\gamma,0), \mmap \oplus \multi{
(\gamma'',j'')} \rangle$, where $j' \leq k$ by hypothesis, giving us
the run $\rho$. 
	\end{proof}

\subsection{Details of the $\DCPS$ simulating a $\TDPN$}
We will now give a polynomial time reduction of the coverability
problem for a transducer defined
Petri net to $\SRP[1]$ for $\DCPS$. We will infact construct a
kill-$\DCPS$, but as shown by Proposition \ref{prop:kill_simulation},
this suffices.

Let the transducer defined Petri net given be $\transpn=(w_{
\mathit{init}},w_{\mathit{final}},\tmove,\tjoin,\tfork)$ with the
three transducers sharing the
common alphabet $\Sigma$ and $w_{\mathit{final}}$ be the target state. Let the explicit Petri net
corresponding to
$\transpn$ be $N=(P,T,F,m_0)$ with target marking $m_f$, where $m_0$ 
(resp. $m_f$) is the
marking with one
token on $w_{\mathit{init}}$ (resp. $w_{\mathit{final}}$) and 0 tokens elsewhere. Let
$l=|w_{\mathit{init}}|$ be the length of
each word in $P$. Let $\tmove=
(\Sigma,\qmove,\initmove,\finmove,\dmove)$, $\tjoin=
(\Sigma,\qjoin,\initjoin,\finjoin,\djoin)$ and $\tfork=
(\Sigma,\qfork,\initfork,\finfork,\dfork)$.

In our description of the kill-$\DCPS$, we will use the placeholder
symbol $\Box$
with subscripts such as $\Box_1,\Box_2$ etc. We will then define the
combination of values tuples of these placeholders could take.

The kill-$\DCPS$ $\cA(\transpn)=(G,\Gamma,\Delta,g_0,\gamma_0)$ is
given as
follows:
\begin{itemize}
	\item $G=G_{\mathit{init}}\cup
	G_{\mathit{check}} \cup G_{\mathit{read}} \cup G_{\mathit{guess}} \cup G_{\mathit{verify}}$ , where
	\begin{itemize}
		\item $G_{\mathit{init}}=\{init\} \times \{ 1, \cdots, l\}$,
		\item $G_{\mathit{check}}=\{ g_{\mathit{main}},g_{\mathit{halt}}\} \cup \{ g_{\mathit{check}}\} \times
		\{
		0,1,\cdots,l\} \cup \{ g_{\mathit{guess}}\} \times \{ move,join,fork\}$,
		\item $G_{\mathit{read}}=\{move,fork,join\} \times \{read,unlock\} \times 
		\{1,\cdots,l\} \times \{
		pop1,pop2\}$,
		\item $G_{\mathit{guess}}=\{move,fork,join\} \times \{1, \cdots, l, toplock\}
		\times
		\{ push1,push2\}$,
		\item $G_{\mathit{verify}}=\{ \qmove \cup \qjoin \cup \qfork\} \times \{1,
		\cdots , l \} \times \{ pop1, pop2, push1, push2\} \times \{
		\dmove \cup \djoin \cup \dfork \cup \Sigma \} \}$,
		
	\end{itemize}
	\item $\Gamma = \Sigma \cup \{ \top, \gamma_{\mathit{guess}} \} \cup \Gamma_
	{kill}$ where \\
	$\Gamma_{\mathit{kill}}=\gamma_{\mathit{verify}} \cup \Sigma \times \{ 1, \cdots ,l\} \times \{
	pop1,pop2,push1,push2\}$.
	\item $g_0=(init,l)$.
	\item $\gamma_0=w_l$ where $w_{\mathit{init}}=w_1w_2 \cdots w_l$.
	
	\item $\Delta$ contains the following rules:
	\begin{itemize}
		
		\item Related to $G_{\mathit{init}}$:
		\begin{enumerate}
			\item $(init,i)|w_i \hookrightarrow (init,i-1)|w_{i-1}w_i$ for all
			$i \in \{ 2, \cdots, l\}$ and
			\item $(init,1)|w_1 \hookrightarrow g_{\mathit{main}}|\top w_1$
		\end{enumerate}
		
		\item Related to $G_{\mathit{check}}$:
		\begin{enumerate}
			\item For each $\Box_1 \in \{ move,join,fork\}$ :\[ g_{\mathit{main}}|\top
			\hookrightarrow (\Box_1,unlock,1,pop1)|\top\]
			\item \[ g_{\mathit{main}}|\top \hookrightarrow (g_{\mathit{check}},1)|\varepsilon \]
			
			\item For each $i \leq l-1$:
			\begin{flalign*}
			(g_{\mathit{check}}&,i)|b_i  \hookrightarrow (g_{\mathit{check}},i+1)|\varepsilon\\
			&\text{ where } w_{\mathit{final}}=b_1b_2...b_l
			\end{flalign*}
			\item 
			\begin{flalign*}
			(g_{\mathit{check}}&,l)|b_l  \hookrightarrow g_{\mathit{halt}}|\varepsilon\\
			&\text{ where } w_{\mathit{final}}=b_1b_2...b_l
			\end{flalign*}
			\item For each $a \in \Sigma$
			\begin{flalign*}
				(g_{\mathit{guess}},&\Box_1)|a  \hookrightarrow (\Box_1,l,push1)|\varepsilon
				\triangleright \gamma_{\mathit{guess}} \\
				&\text{where } \Box_1 \in \{ join,fork,move\}
			\end{flalign*}
		\end{enumerate}
		
		\item Related to $G_{\mathit{read}}$:
		\begin{enumerate}
			\item 
			\begin{flalign*}
			(\Box_1,&unlock, 1,\Box_2)| \top \hookrightarrow
			(\Box_1,read,1,\Box_2)| \varepsilon \\
			\text{where }& (\Box_1,\Box_2) \in \{ (move, pop1),(join,pop1),
			(fork,pop1),(join,pop2)\}
			\end{flalign*}
			\item For each $a \in \Sigma$ and $i \in \{1, \cdots, l-1 \}$ :
			\begin{flalign*}
				(\Box_1,&read,i,\Box_2)|a  \hookrightarrow 
				(\Box_1,read,i+1,\Box_2)|\varepsilon \triangleright (a,i,\Box_2) \\
				 \text{where}& (\Box_1,\Box_2) \in \{ (move,pop1),(join,pop1),
				 (join,pop2),(fork,pop1) \} 
			\end{flalign*}
			\item For each $a \in \Sigma$: 
			
			\begin{flalign*}
				(\Box_1,&read,l,\Box_2)|a  \hookrightarrow 
				(g_{\mathit{guess}},\Box_1)|a
				\triangleright (a,l,\Box_2)\\
				\text{where}& (\Box_1,\Box_2) \in \{ (move,pop1), 
				(join,pop2),(fork,pop1)\}
			\end{flalign*}
			\item For each $a \in \Sigma$: 
			\begin{flalign*}
				(join,&read,l,pop1)|a  \hookrightarrow 
				(join,unlock,1,pop2)|\varepsilon
				\triangleright (a,l,pop1)
			\end{flalign*}
			
		\end{enumerate}
		
		\item Related to $G_{\mathit{guess}}$: 
		\begin{enumerate}
			\item For each $a \in \Sigma$:
			\begin{flalign*}
			(\Box_1,&l,\Box_2)|\gamma_{\mathit{guess}}  \hookrightarrow 
			(\Box_1,l-1,\Box_2)|a
				\triangleright (a,l,\Box_2)\\
				\text{where}& (\Box_1,\Box_2) \in \{ (move,push1),(join,push1),
				(fork,push2),(fork,push1)\}
				\end{flalign*}			
			\item For each $a,b \in \Sigma, 1< i < l$:
			\begin{flalign*}
				(\Box_1,&i,\Box_2)|a  \hookrightarrow (\Box_1,i-1,\Box_2)|b.a
				\triangleright (b,i,\Box_2)\\
				\text{where}& (\Box_1,\Box_2) \in \{ (move,push1),(join,push1),
				(fork,push2),(fork,push1)\}
			\end{flalign*}		
				\item For each $a,b \in \Sigma$:
				\begin{flalign*}
			(\Box_1,&1,\Box_2)|a  \hookrightarrow 
			(\Box_1,toplock,\Box_2)|b.a
				\triangleright (b,1,\Box_2)\\
				\text{where}& (\Box_1,\Box_2) \in \{ (move,push1),(join,push1),
				(fork,push2),(fork,push1)\}
				\end{flalign*}
				\item For each $a \in \Sigma$:
				\begin{flalign*}
			(\Box_1,&toplock,\Box_2)|a  \hookrightarrow 
			\Box_3|\top.a \triangleright \gamma_{\mathit{verify}}\\
				\text{where }& (\Box_1,\Box_2) \in \{ (move,push1),(join,push1),
				(fork,push2)\} \text{ and }\\
				&\Box_3 \in G_{\mathit{verify}}  \text{ has the form } (
				\Box_4,1,pop1,\delta) \text{ with
				}\\ 
				& (\Box_1,\Box_4) \in \{ (\mathit{fork},q_{\mathit{fork0}}),
				(\mathit{join},q_{\mathit{join0}}), (\mathit{move},q_{
				\mathit{move0}})\}
				\text{and } \\
				&\delta \text{ is a
				transition with source } \Box_4
				\end{flalign*}
				
				\item For each $a \in \Sigma$:
				\begin{flalign*}
			(fork,&toplock,push1)|a  \hookrightarrow 
			(fork,l,push2)|\top.a \triangleright \gamma_{\mathit{guess}}
				\end{flalign*}
		\end{enumerate}
		
		\item Related to $G_{\mathit{verify}}$. 
		\begin{enumerate}
			
			\item For each $q \in \qmove, i \leq l-1, \delta=(q  \xrightarrow{
			(a_1,a_2)} q') \in \dmove$:
			\begin{flalign*}
				(q,i,pop1,\delta)|
				\gamma_{\mathit{verify}}  \hookrightarrow
				(q,i,push1,\delta)|\gamma_{\mathit{verify}} \not
				\triangleright (a_1,i,pop1)
			\end{flalign*}
			\item For each $q \in \qmove, i \leq l-1, \delta=(q  \xrightarrow{
			(a_1,a_2)} q'), \delta'=(q'  \xrightarrow{
			(a_1',a_2')} q'') \in \dmove$:
			\begin{flalign*}
				(q,i,push1,\delta)|
				\gamma_{\mathit{verify}} & \hookrightarrow
				(q',i+1,pop1,\delta')|\gamma_{\mathit{verify}} \not
				\triangleright (a_2,i,push1)
			\end{flalign*}
			\item For each $q \in \qmove, \delta=(q \xrightarrow{(a_1,a_2)} q')
			\in \dmove \text{ with } q' \in
			\finmove$:
			\begin{flalign*}
				(q,l,pop1,\delta)|
				\gamma_{\mathit{verify}} & \hookrightarrow
				(q,l,push1,\delta)|\gamma_{\mathit{verify}}
				\not
				\triangleright (a_1,l,pop1)
			\end{flalign*}
			\item For each $q \in \qmove, \delta=(q \xrightarrow{(a_1,a_2)} q')
			\in \dmove \text{ with } q' \in
			\finmove$:
			\begin{flalign*}
				(q,l,push1,\delta)|
				\gamma_{\mathit{verify}} & \hookrightarrow
				g_{\mathit{main}}|\varepsilon
				\not
				\triangleright (a_2,l,push1)
			\end{flalign*}

			\item For each $q \in \qjoin, i \leq l-1, \delta=(q \xrightarrow{
			(a_1,a_2,a_3)} q') \in \djoin$:
			\begin{flalign*}
				(q,i,pop1,\delta)|
				\gamma_{\mathit{verify}} & \hookrightarrow
				(q,i,pop2,\delta)|\gamma_{\mathit{verify}} \not
				\triangleright (a_1,i,pop1)
			\end{flalign*}

			\item For each $q \in \qjoin, i \leq l-1, \delta=(q \xrightarrow{
			(a_1,a_2,a_3)} q') \in \djoin$:
			\begin{flalign*}
				(q,i,pop2,\delta)|
				\gamma_{\mathit{verify}} & \hookrightarrow
				(q,i,push1,\delta)|\gamma_{\mathit{verify}} \not
				\triangleright (a_2,i,pop2)
			\end{flalign*}

			\item For each $q \in \qjoin, i \leq l-1, \delta=(q \xrightarrow{
			(a_1,a_2,a_3)} q')$, \\ $\delta'=$ $(q'$ $\xrightarrow{
			(a_1,a_2,a_3)}q'')$ $
			\in \djoin$:
			\begin{flalign*}
				(q,i,push1,\delta)|
				\gamma_{\mathit{verify}} & \hookrightarrow
				(q',i+1,pop1,\delta')|\gamma_{\mathit{verify}} \not
				\triangleright (a_3,i,push1)
			\end{flalign*}

			\item For each $q \in \qjoin, \delta=(q \xrightarrow{(a_1,a_2,a_3)}
			q') \in \djoin \text{ with } q' \in
			\finjoin$:
			\begin{flalign*}
				(q,l,pop1,\delta)|
				\gamma_{\mathit{verify}} & \hookrightarrow
				(q,l,pop2,\delta)|\gamma_{\mathit{verify}}
				\not
				\triangleright (a_1,l,pop1)
			\end{flalign*}

			\item For each $q \in \qjoin, \delta=(q \xrightarrow{(a_1,a_2,a_3)}
			q') \in \djoin \text{ with } q' \in
			\finjoin$:
			\begin{flalign*}
				(q,l,pop2,\delta)|
				\gamma_{\mathit{verify}} & \hookrightarrow
				(q,l,push1,\delta)|\gamma_{\mathit{verify}}
				\not
				\triangleright (a_2,l,pop2)
			\end{flalign*}

			\item For each $q \in \qjoin, \delta=(q \xrightarrow{(a_1,a_2,a_3)}
			q') \in \djoin \text{ with } q' \in
			\finjoin$:
			\begin{flalign*}
				(q,l,push1,\delta)|
				\gamma_{\mathit{verify}} & \hookrightarrow
				g_{\mathit{main}}|\varepsilon
				\not
				\triangleright (a_3,l,push1)
			\end{flalign*}
			
			\item For each $q \in \qfork, i \leq l-1, \delta=(q \xrightarrow{
			(a_1,a_2,a_3)} q') \in \dfork$:
			\begin{flalign*}
				(q,i,pop1,\delta)|
				\gamma_{\mathit{verify}} & \hookrightarrow
				(q,i,push1,\delta)|\gamma_{\mathit{verify}} \not
				\triangleright (a_1,i,pop1)
			\end{flalign*}

			\item For each $q \in \qfork, i \leq l-1, \delta=(q \xrightarrow{
			(a_1,a_2,a_3)} q') \in \dfork$:
			\begin{flalign*}
				(q,i,push1,\delta)|
				\gamma_{\mathit{verify}} & \hookrightarrow
				(q,i,push2,\delta)|\gamma_{\mathit{verify}} \not
				\triangleright (a_2,i,push1)
			\end{flalign*}

			\item For each $q \in \qfork, i \leq l-1, \delta=(q \xrightarrow{
			(a_1,a_2,a_3)} q')$,\\ $\delta'=(q' \xrightarrow{(a_1',a_2',a_3')}
			q'')
			\in \dfork$:
			\begin{flalign*}
				(q,i,push2,\delta)|
				\gamma_{\mathit{verify}} & \hookrightarrow
				(q',i+1,pop1,\delta')|\gamma_{\mathit{verify}} \not
				\triangleright (a_3,i,push2)
			\end{flalign*}

			\item For each $q \in \qfork, \delta=(q \xrightarrow{(a_1,a_2,a_3)}
			q') \in \dfork \text{ with } q' \in
			\finfork$:
			\begin{flalign*}
				(q,l,pop1,\delta) q')|
				\gamma_{\mathit{verify}} & \hookrightarrow
				(q,l,push1,\delta)|\gamma_{\mathit{verify}}
				\not
				\triangleright (a_1,l,pop1)
			\end{flalign*}

			\item For each $q \in \qfork, \delta=(q \xrightarrow{(a_1,a_2,a_3)}
			q') \in \dfork \text{ with } q' \in
			\finfork$:
			\begin{flalign*}
				(q,l,push1,\delta)|
				\gamma_{\mathit{verify}} & \hookrightarrow
				(q,l,push2,\delta)|\gamma_{\mathit{verify}}
				\not
				\triangleright (a_2,l,push1)
			\end{flalign*}

			\item For each $q \in \qfork, \delta=(q \xrightarrow{(a_1,a_2,a_3)}
			q') \in \dfork \text{ with } q' \in
			\finfork$:
			\begin{flalign*}
				(q,l,push2,\delta)|
				\gamma_{\mathit{verify}} & \hookrightarrow g_{\mathit{main}}|\varepsilon
				\not
				\triangleright (a_3,l,push2)
			\end{flalign*}
		\end{enumerate}

	\end{itemize}
	
\end{itemize}

\begin{lemma}
	\label{lem:transpn_to_dcps}
The marking $\mmap_f$ is coverable in $N$ iff $g_{\mathit{halt}}$ is
1-bounded
reachable in $\cA(\transpn)$.
\end{lemma}
\begin{proof}
	We will prove the following statement towards establishing the
	equivalence:\\
\underline{	\textbf{Claim:}} A marking $\mmap_k$ of $N(\cN)$ can be
reached
in
$N(\cN)$ by using
	$k \geq 1$ transducer transitions iff for every place $w' \in P$ such
	that
	$\mmap_k(w') \geq 1$ there exists a 1-bounded run $\rho$ of $\cA
	(\transpn)$ such that $\rho=\langle g_0,(\gamma_0,0),\emptyset
	\rangle
	\Rightarrow^*_{\leq 1} \langle g_{\mathit{main}}, (\top w',1),\mmap \rangle$
	where
	the
	following
	conditions hold:
	\begin{itemize}
		\item $g_{\mathit{main}}$ occurs $k+1$ times in $\rho$,
		\item for each $u \in P$, $\mmap \oplus \multi{(\top w',1)}$ contains
		exactly $\mmap'(u)$ threads each of which have as stack content
		$\top
		u$ and context switch number $1$ and
		\item $\mmap \oplus \multi{(\top w',1)}$ does not contain any other
		threads than those in the above point. 
		\item The context switch number of any thread in any configuration
		of $\rho$ is at most 1.
	\end{itemize}

We will prove the claim by induction on the number $k$. 
	The initial marking $\mmap_0$ is 
	 $\mmap_0 =\multi{w_{\mathit{init}}}$.
	Let $w_{\mathit{init}}=a_1a_2..a_l$. It is clear that the rules related to $G_{\mathit{init}}$ can be used to reach
	$\langle g_{\mathit{main}},(\top a_1 a_2...a_l,0),\emptyset\rangle$ as follows:
	\begin{flalign*}
		&\langle (init,l),(w_l,0),\emptyset\rangle \\
		\xrightarrow{G_{\mathit{init}}.1}& \langle (init,l-1),(a_{l-1}a_l),\emptyset\rangle\\
		& \vdots \\
		\xrightarrow{G_{\mathit{init}}.1}& \langle (init,1),(a_1...a_{l-1}a_l),\emptyset\rangle\\
		\xrightarrow{G_{\mathit{init}}.2}& \langle g_{\mathit{main}},(\top a_1...a_l,0),\emptyset\rangle
	\end{flalign*}
Note that the above is the unique run which reaches $g_{\mathit{main}}$ for the
first time (since there
is only one thread).

The base case where $k=1$ is similar to the induction step explained
below. We point out the small ways in which the base case differs
after proving the induction step. 

\noindent \underline{\textit{Induction Step:}} Let $\mmap_k$ be a
marking
reached
in $N(\cN)$ via
$k$ transducer-moves:
\[ \mmap_0 \xrightarrow{t_1} \mmap_1 \xrightarrow{
t_2} \mmap_2
\cdots \xrightarrow{t_{k-1}} \mmap_{k-1} \xrightarrow{
t_k} \mmap_k\]
Corresponding to each transition $t_i$ there exists an accepting run
of a transducer $\mathcal{T}_i$ which is one of $\{
\tmove,\tjoin,\tfork\}$, we call $t_i$ a $
\mathcal{T}_i$-move. The $
\mathcal{T}_i$-move is a
sequence of transitions $\delta_{i1} \cdots \delta_{il}$ from the
corresponding transducer.

By induction hypothesis, for any $u \in P$ with $\mmap_{k-1}(u) \geq
1$ there exists a run $\rho'$ of $\cA(\transpn)$
in which $g_{\mathit{main}}$ occurs $k$ times and $\rho'=\langle g_0,(\gamma_0,0),
\emptyset\rangle \Rightarrow^*_{\leq 1} \langle g_{
\mathit{main}},(u,1),\mmap\rangle$ and for
all $w
\in P$ there are $\mmap_{k-1}(w)$ threads in $\mmap \oplus \multi{
(u,1)})$ each
of which has stack content $\top w$. 

Let $\mmap_{k-1} \xrightarrow{t_k} \mmap_k$ be a $\tjoin$-move. Then
there exist $w',w'',w'''$
such that \linebreak $\initjoin \xrightarrow{(w',w'',w''')} q_{\mathit{joinf}}$
where $q_{\mathit{joinf}} \in \finjoin$. This implies $\mmap_k=\mmap_
{k-1} \oplus \multi{w'''}
\ominus \multi{w',w''}$.

By induction hypothesis, we can assume that $u=\top w'$ and there
exists a thread $t$ in $\mmap$ with stack content $\top w''$. Let
$\mmap' \oplus \multi{(\top w'',1)}= \mmap$.

Let $w'=w'_1 \cdots w'_l, w''=w''_1 \cdots w''_l, w'''=w'''_1 \cdots
w'''_l$. We outline the simulation of a transducer-move by $\cA
(\transpn)$
below.

\subsection{Simulation of a Single Transducer-move}
\label{sim-trans-move}
We label each transition by either the corresponding rule if it is a
transition of type $\rightarrow$ (for instance $
\xrightarrow{G_{\mathit{read}}.1}_1$ means that the $1^{st}$ rule related to
$G_{\mathit{read}}$ is being used ) or label it by the $c.s. i$ if it is
the $i^{th}$ context switch ( for example, $\xmapsto{c.s. i}_1$).
\begin{enumerate}
	\item Starting from the configuration $\langle g_{\mathit{main}},(\top
	w',1),\mmap' \oplus \multi{(\top w'',1)}\rangle$, we use the rule
 $G_{\mathit{check}}.1$ to unlock the thread $(\top
	w',1)$ and move to \newline $\langle (join,read,1,pop1),(w',1),
	\mmap'\oplus
	\multi{(\top w'',1)} \rangle$
	\begin{flalign*}
		&\langle g_{\mathit{main}},(\top w',1),\mmap' \oplus \multi{(\top w'',1)}\rangle \\
	\xrightarrow{G_{\mathit{check}}.1}_1& \langle (join,unlock,1,pop1),(\top
	w',1),\mmap' \oplus \multi{(\top w'',1)}\rangle
	\end{flalign*}
 \item We then use the rules related to $G_{\mathit{read}}$ to pop off the
 stack contents of $(w',1)$ symbol by symbol while simultaneously
 spawning threads of the kind $(w'_i,i,pop1)$, arriving at the
 configuration $\langle (join,unlock,1,pop2),(\varepsilon,1), \mmap'
 \oplus
	\multi{(\top w'',1)} \oplus_{i=1}^{l} \multi{
	((w'_i,i,pop1),0)}\rangle$.
	\begin{flalign*}
		& \langle (join,unlock,1,pop1),(\top
	w',1),\mmap' \oplus \multi{(\top w'',1)}\rangle\\
	\xrightarrow{G_{\mathit{read}}.1}_1& \langle (join,read,1,pop1),(
	w',1),\mmap' \oplus \multi{(\top w'',1)}\rangle\\
		\xrightarrow{G_{\mathit{read}}.2}_1& \langle (join,read,2,pop1),(
	w'_2 \cdots w'_l,1),\mmap' \oplus \multi{(\top w'',1)} \oplus \multi{
	((w'_1,1,pop1),0)}\rangle\\
	&\vdots\\
	\xrightarrow{G_{\mathit{read}}.2}_1& \langle (join,read,l,pop1),(
	w'_l,1),\mmap' \oplus \multi{(\top w'',1)} \oplus_{i=1}^{l-1} \multi{
	((w'_i,i,pop1),0)}\rangle \\
	\xrightarrow{G_{\mathit{read}}.4}_1& \tuple{(join,unlock,1,pop2),
	(\varepsilon,1),\mmap'\oplus \multi{(\top w'',1)} \oplus_{i=1}^{l} \multi{
	((w'_i,i,pop1),0)}}
	\end{flalign*}
 \item Next, we context switch to $(\top w'',1)$ and repeat the
 process of unlocking, poping the stack contents and simultaneously
 spawning threads $(w''_i,i,pop2)$. 
 \begin{flalign*}
 	&\tuple{(join,unlock,1,pop2),
	(\varepsilon,1),\mmap'\oplus \multi{(\top w'',1)} \oplus_{i=1}^{l} \multi{
	((w'_i,i,pop1),0)}}\\	
	\xmapsto{c.s. 1}_1 & \tuple{(join,unlock,1,pop2),(\top w'',1),\mmap'
	\oplus_{i=1}^{l} \multi{
	((w'_i,i,pop1),0)}}\\
	\xrightarrow{G_{\mathit{read}}.1}_1& \langle (join,read,1,pop2),(
	w'',1),\mmap' \oplus_{i=1}^{l} \multi{
	((w'_i,i,pop1),0)}\rangle\\
		\xrightarrow{G_{\mathit{read}}.2}_1& \langle (join,read,2,pop2),(
	w''_2 \cdots w'_l,1),\\
	&\mmap' \oplus_{i=1}^{l} \multi{
	((w'_i,i,pop1),0)} \oplus \multi{
	(w''_1,1,pop2)}\rangle\\
	&\vdots\\
	\xrightarrow{G_{\mathit{read}}.2}_1& \langle (join,read,l,pop2),(
	w''_l,1),\\
	&\mmap' \oplus_{i=1}^{l} \multi{
	((w'_i,i,pop1),0)} \oplus_{i=1}^{l-1} \multi{
	((w''_i,i,pop2),0)}\rangle
 \end{flalign*}
 \item While reading the last (i.e. $l^{th}$) letter of $w''$, we
 move to the state $(g_{\mathit{guess}},join)$ without poping. 
 \begin{flalign*}
 	& \langle (join,read,l,pop2),(
	w''_l,1), \\
	&\mmap' \oplus_{i=1}^{l} \multi{
	((w'_i,i,pop1),0)} \oplus_{i=1}^{l-1} \multi{
	((w''_i,i,pop2),0)}\rangle \\
	\xrightarrow{G_{\mathit{read}}.3}_1& \langle(g_{\mathit{guess}},join),
	(w''_l,1), \\
	&\mmap' \oplus_{i=1}^{l} \multi{
	((w'_i,i,pop1),0)}\oplus_{i=1}^{l} \multi{
	((w''_i,i,pop2),0)}\rangle
 \end{flalign*}
	\item We then use the rule $G_{\mathit{check}}.5$ to pop $w''_l$ and spawn a
	thread $\gamma_{\mathit{guess}}$. 
	\begin{flalign*}
		&\langle(g_{\mathit{guess}},join),
	(w''_l,1),\\
	&\mmap' \oplus_{i=1}^{l} \multi{
	((w'_i,i,pop1),0)}\oplus_{i=1}^{l} \multi{
	((w''_i,i,pop2),0)}\rangle\\
	\xrightarrow{G_{\mathit{check}}.5}_1& \langle(join,l,push1),
	(\varepsilon,1), \\
	&\mmap' \oplus_{i=1}^{l} \multi{
	((w'_i,i,pop1),0)}\oplus_{i=1}^{l} \multi{
	((w''_i,i,pop2),0)} \oplus \multi{(\gamma_{\mathit{guess}},0)} \rangle
	\end{flalign*}
		\item Context switching to the thread $\gamma_{\mathit{guess}}$, the rules
		related to $G_{\mathit{guess}}$ are used to guess a thread
		with stack contents $(\top w''',1)$ while simultaneously spawning
		threads $(w'''_i,$ $i,$ $push1)$ for each $i \leq l$. At this point we
		have threads (one per each letter) which contain the information regarding
		$w',w'',w'''$. At the end of the guess process, we spawn a thread
		$\gamma_{\mathit{verify}}$.
		\begin{flalign*}
	&\langle(join,l,push1),
	(\varepsilon,1), \\
	&\mmap' \oplus_{i=1}^{l} \multi{
	((w'_i,i,pop1),0)}\oplus_{i=1}^{l} \multi{
	((w''_i,i,pop2),0)} \oplus \multi{(\gamma_{\mathit{guess}},0)} \rangle\\
	\xmapsto{c.s. 2}_1& \langle (join,l,push1),
	(\gamma_{\mathit{guess}},0), \\
  &\mmap' \oplus_{i=1}^{l} \multi{
	((w'_i,i,pop1),0)}\oplus_{i=1}^{l} \multi{
	((w''_i,i,pop2),0)} \rangle\\
	\xrightarrow{G_{\mathit{guess}}.1}_0& \langle(join,l-1,push1),
	(w'''_l,0),\\
	&\mmap' \oplus_{i=1}^{l} \multi{
	((w'_i,i,pop1),0)}\oplus_{i=1}^{l} \multi{
	((w''_i,i,pop2),0)} \oplus \multi{((w'''_l,l,push1),0)} \rangle\\
	\xrightarrow{G_{\mathit{guess}}.2}_0& \langle (join,l-2,push1),
	(w'''_{l-1}w'''_l,0), \\
	&\mmap' \oplus_{i=1}^{l} \multi{((w'_i,i,pop1),0)}\oplus_{i=1}^{l} \multi{
	((w''_i,i,pop2),0)} \\
  &\oplus_{i=l-1}^l \multi{((w'''_i,i,push1),0)} \rangle\\
	&\vdots	\\
	\xrightarrow{G_{\mathit{guess}}.2}_0& \langle (join,1,push1),
	(w'''_2 \cdots w'''_{l-1}w'''_l,0), \\
	&\mmap' \oplus_{i=1}^{l} \multi{((w'_i,i,pop1),0)}\oplus_{i=1}^{l} \multi{
	((w''_i,i,pop2),0)} \\
  &\oplus_{i=2}^l \multi{((w'''_i,i,push1),0)} \rangle\\
	\xrightarrow{G_{\mathit{guess}}.3}_0& \langle(join,toplock,push1),
	(w'''_1 \cdots w'''_{l-1}w'''_l,0),\\
	&\mmap' \oplus_{i=1}^{l} \multi{((w'_i,i,pop1),0)}\oplus_{i=1}^{l} \multi{
	((w''_i,i,pop2),0)} \\
  &\oplus_{i=1}^l \multi{((w'''_i,i,push1),0)} \rangle\\
		\xrightarrow{G_{\mathit{guess}}.4}_0& \langle(\initjoin,1,pop1,\delta_0),
	(\top w'''_1 \cdots w'''_{l-1}w'''_l,0),\\
	&\mmap' \oplus_{i=1}^{l} \multi{((w'_i,i,pop1),0)}\oplus_{i=1}^{l} \multi{
	((w''_i,i,pop2),0)} \oplus_{i=1}^l \multi{((w'''_i,i,push1),0)} \\
	& \oplus \multi{(\gamma_{\mathit{verify}},0)}\rangle
\end{flalign*}

		\item We switch to the thread $\gamma_{\mathit{verify}}$ and begin the
		process of verifying that our guess of $w'''$ enables us to make a
		valid accepting run of the join transducer. This is a
		sequence of transitions $(\initjoin
		\xrightarrow{
	\delta_{k,1}} q_
	{k,1} \xrightarrow{\delta_{k,2}} \cdots \xrightarrow{\delta_{k,l}} q_
	{k,l})$ where $q_{k,l} \in \finjoin$. This is accomplished by killing
	the threads \newline 
	$(w',i,pop1),(w'',i,pop2),(w''',i,push1)$ for
	each $i
	\leq l$ while checking that $\delta_{k,i} =$\linebreak$(q_{k,i-1}
		\xrightarrow{
	(w'_i,w''_i,w'''_i)} q_
	{k,i})$. 
	\begin{flalign*}
	& \langle(\initjoin,1,pop1,\delta_{k,1}),
	(\top w'''_1 \cdots w'''_{l-1}w'''_l,0),\\
	&\mmap' \oplus_{i=1}^{l} \multi{((w'_i,i,pop1),0)}\oplus_{i=1}^{l} \multi{
	((w''_i,i,pop2),0)} \\
  &\oplus_{i=1}^l \multi{((w'''_i,i,push1),0)} \oplus \multi{(\gamma_{\mathit{verify}},0)}\rangle \\
	\xmapsto{c.s. 3}_0& \langle(\initjoin,1,pop1,\delta_{k,1}),
	(\gamma_{\mathit{verify}},0),\\
	&\mmap' \oplus_{i=1}^{l} \multi{((w'_i,i,pop1),0)}\oplus_{i=1}^{l} \multi{
	((w''_i,i,pop2),0)} \\
  &\oplus_{i=1}^l \multi{((w'''_i,i,push1),0)} \oplus\multi{(\top
	w''',1)}\rangle
	\\
	&\text{where } \delta_{k,1}=(\initjoin \xrightarrow{
	(w'_1,w''_1,w'''_1)} q_
	{k,1}) \in \djoin\\
	\xrightarrow{G_{\mathit{verify}}.5}_0& \langle(\initjoin,1,pop2,\delta_{k,1}),
	(\gamma_{\mathit{verify}},0),\\
	&\mmap' \oplus_{i=2}^{l} \multi{((w'_i,i,pop1),0)}\oplus_{i=1}^{l} \multi{
	((w''_i,i,pop2),0)} \\
  &\oplus_{i=1}^l \multi{((w'''_i,i,push1),0)} \oplus \multi{(\top
	w''',1)}\rangle\\
	\xrightarrow{G_{\mathit{verify}}.6}_0& \langle(\initjoin,1,push1,\delta_{k,1}),
	(\gamma_{\mathit{verify}},0),\\
	&\mmap' \oplus_{i=2}^{l} \multi{((w'_i,i,pop1),0)}\oplus_{i=2}^{l} \multi{
	((w''_i,i,pop2),0)} \\
  &\oplus_{i=1}^l \multi{((w'''_i,i,push1),0)} \oplus\multi{(\top
	w''',1)}\rangle\\
	\xrightarrow{G_{\mathit{verify}}.7}_0& \langle(q_{k,1},2,pop1,\delta_{k,2}),
	(\gamma_{\mathit{verify}},0),\\
	&\mmap' \oplus_{i=2}^{l} \multi{((w'_i,i,pop1),0)}\oplus_{i=2}^{l} \multi{
	((w''_i,i,pop2),0)} \\
  &\oplus_{i=2}^l \multi{((w'''_i,i,push1),0)} \oplus\multi{(\top
	w''',1)}\rangle\\
	&\text{where } \delta_{k,2}=(q_{k,1} \xrightarrow{
	(w'_2,w''_2,w'''_2)} q_{k,2})\\
	&\vdots\\
	\xrightarrow{G_{\mathit{verify}}.7}_0& \langle(q_{k,l},l,pop1,\delta_{k,l}),
	(\gamma_{\mathit{verify}},0),\\
	&\mmap' \oplus \multi{((w'_l,l,pop1),0)} \oplus \multi{
	((w''_l,l,pop2),0)} \oplus \multi{((w'''_l,l,push1),0)} \\
	&\oplus\multi{(\top
	w''',1)}\rangle\\
	&\text{ where } \delta_{k,l}=(q_{k,l-1} \xrightarrow{
	(w'_l,w''_l,w'''_l)} q_
	{k,l}) \in \djoin \text{ and } q_{k,l} \in \finjoin\\
	\xrightarrow{G_{\mathit{verify}}.8}_0& \langle(q_{k,l},l,pop2,\delta_{k,l}),
	(\gamma_{\mathit{verify}},0),\\
	&\mmap' \oplus \multi{
	((w''_l,l,pop2),0)} \oplus \multi{((w'''_l,l,push1),0)} \\
	&\oplus\multi{(\top
	w''',1)}\rangle\\
	\xrightarrow{G_{\mathit{verify}}.9}_0& \langle(q_{k,l},l,push1,\delta_{k,l}),
	(\gamma_{\mathit{verify}},0),\\
	&\mmap' \oplus \multi{((w'''_l,l,push1),0)} \oplus\multi{(\top
	w''',1)}\rangle\\
	\xrightarrow{G_{\mathit{verify}}.10}_0& \langle g_{\mathit{main}},
	(\varepsilon,0),\mmap' \oplus\multi{(\top
	w''',1)}\rangle\\
	\xmapsto{c.s. 4}_0 & \langle g_{\mathit{main}},(v,1),\mmap''\rangle\\
	&\text{ where } \mmap'' \oplus \multi{(v,1)} = \mmap' \oplus \multi{(\top
	w''',1)}
\end{flalign*}
	\item At this point, we can make a context switch to any of the
	threads in $\mmap' \oplus\multi{(\top
	w''',1)}$ to start the simulation of the next transducer-move.
	\begin{flalign*}
		&\langle g_{\mathit{main}},
	(\varepsilon,0),\mmap' \oplus\multi{(\top
	w''',1)}\rangle\\
	\mapsto_0 & \langle g_{\mathit{main}},(v,1),\mmap''\rangle\\
	&\text{ where } \mmap'' \oplus \multi{(v,1)} = \mmap' \oplus \multi{(\top
	w''',1)}
	\end{flalign*}
\end{enumerate}

The cases of the transducer-move belonging to $\tmove$ or $\tfork$ is
similar. The only differences are that in the case of $\tmove$ there
is only one stack to be read before we proceed to the
guess of the stack related to $push1$ and in the
case of $\tfork$ there are two stacks to be guessed. This concludes
the proof that for any marking $\mmap_k$ reached in $k$
transducer-moves
by $N(\cN)$, there exists a 1-bounded run $\rho$ of $\cA(\transpn)$ that
reaches
the state $g_{\mathit{main}}$ and in which $g_{\mathit{main}}$ occurs $k+1$ times.

Conversely, consider a run $\rho$ of $\cA(\transpn)$ which satisfies
the conditions in the claim:
\begin{flalign*} 
&\tuple{g_0,(\gamma_0,0),\emptyset} \\
 \Rightarrow_0^* & \tuple{g_{\mathit{main}},(\top w_{\mathit{init}},0),\emptyset}\\
 \Rightarrow_{\leq 1}^* & \tuple{g_{\mathit{main}},(\top u_1,1),\mmap_1}\\
 &\vdots \\
 \Rightarrow_{\leq 1}^* & \tuple{g_{\mathit{main}},(\top u_{k-1},1),\mmap_
 {k-1}}\\
  \Rightarrow_{\leq 1}^* & \tuple{g_{\mathit{main}},(\top u_k,1),\mmap_k}
\end{flalign*}

Let $\mmap_{k-1}$ be the marking reached by $N(\cN)$ obtained by applying
the
induction hypothesis to the initial segment of $\rho$ which reaches
configuration \newline $\tuple{g_{\mathit{main}},(\top u_{k-1},1),\mmap_{k-1}}$.
Consider the
segment $\rho' = \tuple{g_{\mathit{main}},(\top u_{k-1},1),\mmap_{k-1}}
  \Rightarrow_{\leq 1}^*$ \linebreak $\tuple{g_{\mathit{main}},(\top u_k,1),\mmap_k}$. we
  assume that the first transition in $\rho'$ is to the state
  $(join,unlock,1,$ $pop1)$ (the other two choices being $
  (move,unlock,1,pop1)$ and $(fork,unlock,1,pop1)$, for both of
  which a similar argument holds). By
  construction, $\rho'$ must reach the global states in each of the
  starting
  configurations outlined in the steps 1 through 8 in Subsection 
  \ref{sim-trans-move}. It suffices to argue that the number of
  threads and their contents do not deviate from the configurations in
  steps 2 to 8. First, we observe that we cannot insert any context
  switches after a configuration containing an active thread $(u,1)$
  where $u \neq \varepsilon$. This is because such a switch increases the
  context switch number of the active thread to 2, which is disallowed
  by the conditions placed on $\rho$. Thus the only deviation from the
  run can occur at the following two types of places:
  \begin{enumerate}
   	\item the four context
  switches used in $\rho$ and
  \item configurations where the active thread has context switch
  number 0.
   \end{enumerate}  
  Let us consider the first case above, i.e. at the context switches
  where we
  could have switched to a different thread. There are in total 4 such
  context switches. At $c.s. 1$ we could switch to a different thread
  $(\top \tilde{w}'',1)$ instead of $(\top w'',1)$. However, the
  corresponding run $\tilde{\rho}$ on making this change behaves
  in a similar way in the sense that $(w',\tilde{w}'',w''')$ must
  belong to the language of $\tjoin$ for us to successfully return to
  a configuration with state $g_{\mathit{main}}$. At $c.s. 2$, the only
  possible transition from a state $(\Box_1,l,\Box_2) \in G_{\mathit{guess}}$
  requires the top of stack symbol to be $\gamma_{\mathit{guess}}$. Switching
  to any other thread and switching back leads to an increase in
  context switch number to more than 1, which is disallowed. Similarly
  at $c.s. 3$ we require the top of stack symbol to be $\gamma_{\mathit{verify}}$ and the same argument as before applies. At $c.s. 4$ we
  are allowed to switch to any thread we want and this is intentional.

  Next we consider the second case i.e. places where the active thread
  has context switch number 0. The first place where this occurs is at
  Step 6. At a configuration with state $(join,i,push1)$, the inactive
  threads either have top of stack symbol $\top$ or a symbol from
  $\Gamma_{\mathit{kill}} \setminus \{ \gamma_{\mathit{verify}}\}$. The only transitions
  allowed are a sequence of context switches ending with a switch to
  back to the original active thread. None of the inactive threads with
  top of
  stack $\top$ can be involved in the context switches since this
  would increase such a thread's switch number to at least
  2, which is disallowed by the
  1-boundedness of $\rho$. The only threads which can be involved are
  those with top of 
  stack symbol belonging
  to $\Gamma_{\mathit{kill}} \setminus \{ \gamma_{\mathit{verify}}\}$. Each of these can
be switched to at most once since they are required to be killed later
in the run and this is disallowed if their switch number is strictly
more than 1. Thus the only difference in configuration can be that
such threads have switch number 1 instead of 0. A similar argument as
above applies
to the case
  when the
  active thread is $(\gamma_{\mathit{verify}},0)$ in Step 7. We note that by
  the end of Step 7, all of the threads with a symbol from
  $\Gamma_{\mathit{kill}} \setminus \{ \gamma_{\mathit{verify}}\}$ have been killed
  and thus the configuration must be identical to that obtained
  without any deviation from Steps 1 through 8.\\
  From the above, we can conclude that the run $\rho'$ corresponds to
  the application of a join transducer move on the marking $\mmap_
  {k-1}$ giving us a marking $\mmap_k$ of $N(\cN)$ with the desired
  properties.
  \subparagraph*{Regarding the base case:} The base case differs from
  the
  induction step in that the first time we arrive at $g_{\mathit{main}}$, the
  configuration is $\tuple{g_{\mathit{main}},(\top w_{\mathit{init}},0),\emptyset}$.
  The first transducer-move has to be either from $\tmove$ or $\tfork$
  and when arriving in $g_{\mathit{main}}$ for the second time, all threads
  have switch number 1, as can be seen from the fact that after the
  guess in Step 6, the newly created thread is switched out before the
  verification process in Step 7, resulting in a switch number of 1.
  This concluces our proof of the claim.

We now use the claim to prove the lemma.
Let $\rho'$ be a run of $\transpn$ which reaches a marking $\mmap$
covering $\mmap_f$. In other words, $\mmap(w_{\mathit{final}}) \geq
1$. Suppose $\rho'$
uses $k'$ transducer-moves, then by the claim, there is a run
$\rho$ of $\cA(\transpn)$ such that $\rho=\langle g_0,
(\gamma_0,0),\emptyset\rangle 
	\Rightarrow^*_{\leq 1} \langle g_{\mathit{main}}, (\top w_{\mathit{final}},1),\mmap\rangle$
	 Let $w_{\mathit{final}}=b_1b_2...b_l$. We now have
	the following sequence of transitions which enable us to reach $g_{\mathit{halt}}$:
	\begin{flalign*}
		&\langle g_{\mathit{main}}, (\top w_{\mathit{final}},1),\mmap \rangle \\
		\xrightarrow{G_{\mathit{check}}.2}_1& \langle(g_{\mathit{check}},1), ( w_{\mathit{final}},1),\mmap \rangle \\
		\xrightarrow{G_{\mathit{check}}.3}_1& \langle (g_{\mathit{check}},2), (
		b_2b_3...b_l,1),\mmap \rangle\\
		& \vdots \\
		\xrightarrow{G_{\mathit{check}}.3}_1& \langle (g_{\mathit{check}},l), (b_l,1),\mmap
		\rangle\\
		\xrightarrow{G_{\mathit{check}}.4}_1&  \langle g_{\mathit{halt}}, 
		(\varepsilon,1),\mmap\rangle
	\end{flalign*}

	Conversely, let $\rho$ be a run of $\cA(\transpn)$ that
	reaches
	$g_{\mathit{halt}}$.  Then by construction there is an initial segment
	$\rho''$
	of $\rho$ which reaches $g_{\mathit{main}}$ for the last time before we move
	to
	$(g_{\mathit{check}},1)$ i.e. $\rho$ must be of the following form
\begin{flalign*}
	&\langle g_0,
	(\gamma_0,0),\emptyset\rangle \\
	\Rightarrow^*_{\leq 1} &\langle g_{\mathit{main}},(w_{\mathit{final}},1),\mmap \rangle\\
	\xrightarrow{G_{\mathit{check}}.2}_1& \langle(g_{\mathit{check}},1), ( w_{\mathit{final}},i),\mmap \rangle \\
		\Rightarrow^*_{\leq 1} &\langle g_{\mathit{halt}}, 
		(\varepsilon,1),\mmap\rangle
\end{flalign*}
	 Note that once we reach $(g_{\mathit{check}},1)$,
	all inactive threads have switch number 1 and no context switches are
	possible. Hence there is exactly one sequence of transitions which
	enables us to reach $g_{\mathit{halt}}$ from $(g_{\mathit{check}},1)$. This implies the
	multiset $\mmap$ when we reach $g_{\mathit{halt}}$ is exactly the same as that
	at the configuration $\langle g_{\mathit{main}},(w_{\mathit{final}},1),\mmap \rangle$ reached at the end of $\rho''$.

		Assuming that $\rho''$ has $k'+1$ occurrences of $g_{\mathit{main}}$ and
		applying the claim, we obtain a run of $\transpn$ which uses $k'$
		transducer-moves and ends in a final marking $\mmap_{k'}$ which must
		satisfy $\mmap_{k'}(w_{\mathit{final}}) \geq 1$.
		\end{proof}

\section{Proofs for \cref{sec:lower_bound}}

\subsection{Proofs for \cref{sec:MinskyToRNP}} \label{sec:decinduction}

Our construction borrows heavily from the Lipton construction~\cite{lipton1976reachability}, as it is explained in~\cite{Esp98a}. Therefore most to all of these proofs rely on ideas that were already established there. 
Let us prove the following three statements regarding $R(C)$:
\begin{enumerate}
\item Assume at the start of the execution of \texttt{dec} at recursion depth $d > 1$ we have $v_{d'} = 2^{2^{n+1-d'}}$ and $\bar{v}_{d'} = 0$ for each $v \in \{\bar{s},y,z\}$ and all $d' > d$. Then either $s_d$ gets decremented by $2^{2^{n+1-d}}$ and $\bar{s}_d$ gets incremented by the same amount, or the execution of \texttt{dec} gets stuck.
\item Assume at the start of the execution of Test$_{+1}(v, l_{\text{zero}}, l_{\text{nonzero}})$ with $v \in \{y,z\}$ at recursion depth $d > 1$ we have $u_{d'} = 2^{2^{n+1-d'}}$ and $\bar{u}_{d'} = 0$ for each $u \in \{\bar{s},y,z\}$ and all $d' > d+1$. Furthermore, assume that $s_{d+1} = 0$, $\bar{s}_{d+1} = 2^{2^{n - d}}$, and $v_{d+1} + \bar{v}_{d+1} = 2^{2^{n - d}}$. Then either we go to $l_{\text{nonzero}}$ with no side effects if $v_{d+1} \neq 0$ at the start, or we go to $l_{\text{zero}}$ with the sole side effect of swapping the values of $v_{d+1}$ and $\bar{v}_{d+1}$ if $v_{d+1} = 0$ at the start, or the execution of Test$_{+1}$ gets stuck.
\item Assume at the start of the execution of \texttt{inc} at recursion depth $d > 1$ all variables have value $0$. Then either the execution of \texttt{inc} gets stuck, or for each $v \in \{\bar{s},y,z\}$ and all $d' \geq d$ the variable $v_{d'}$ gets incremented by $2^{2^{n+1-d'}}$.
\end{enumerate}
We prove these using induction on $d$.

Regarding the base case for \texttt{inc} or \texttt{dec}, we consider the maximum recursion depth $d = n + 1$. Thus, \texttt{inc}$_{=\text{max}}$ or \texttt{dec}$_{=\text{max}}$ get called which perform the increments and decrements by $2^{2^{n+1-(n+1)}} = 2^{2^0} = 2^1 = 2$ on the correct variables. For Test$_{+1}$, the base case is $d = n$, because it contains \textbf{call} commands, which do not occur at maximum recursion depth by definition. We prove this base case together with the inductive case:

Regarding the inductive case for Test$_{+1}(v, l_{\text{zero}}, l_{\text{nonzero}})$, we can jump to $l_{\text{nonzero}}$ iff we previously jump to $l_{\text{nztest}}$ and perform a valid nonzero test on $v_{d+1}$ there via decrementing and incrementing once. This then leaves no side effects. On the other hand, we can jump to $l_{\text{zero}}$ iff by induction hypothesis we can decrement $s_{d+1}$ by $2^{2^{n+1-(d+1)}}$ using the call to \texttt{dec} at $l_{\text{exit}}$. Since we assumed $s_{d+1}$ to be $0$ initially, we have to increment it that many times beforehand, requiring that many visits to $l_{\text{loop}}$. This then correctly shifts the value $2^{2^{n - d}}$ from $\bar{v}_{d+1}$ to $v_{d+1}$ which means that we had $v_{d+1} = 0$ at the start from our assumptions.

Regarding the inductive case for \texttt{dec}, if the execution does not get stuck, we have to visit $l_{\text{outer}}$ $2^{2^{n - d}}$-times while visiting $l_{\text{inner}}$ $2^{2^{n - d}}$-times for each visit to $l_{\text{outer}}$. This is because we decrement either $y_{d+1}$ or $z_{d+1}$ once for each visit, both of them had initial value $2^{2^{n - d}}$, and by induction hypothesis Test$_{+1}$ correctly tests them for zero. In the zero case, these variables are also conveniently reset to their maximum value, meaning that there are no side effects on them. All in all, $s_d$ gets decremented by $2^{2^{n - d}} \cdot 2^{2^{n - d}} = 2^{2^{n - d} + 2^{n - d}} = 2^{2 \cdot 2^{n - d}} = 2^{2^{n - d + 1}}$, which is the correct amount.

Regarding the inductive case for \texttt{inc}, the recursive call gives us the statement for each $d' > d$ per the induction hypothesis. The remainder is very similar to \texttt{dec} and relies on the fact that all the variable values required by Test$_{+1}$ have already been set by the recursive call to \texttt{inc}.

\vspace{1ex}

\noindent It is now clear that the call to \texttt{inc} at the start of $R_{\mathit{init}}(C)$ sets up all the requirements so that \texttt{dec} and Test$_{+1}$ can work correctly. The remainder of $R_{\mathit{init}}(C)$ uses the variable $y_1$ to increment $\bar{x}_0$ by $2^{2^n}$ for each $x \in X$. The correctness of the macro Test regarding $x_0$, since it is only used in the main program, can then easily be inferred from the correctness of Test$_{+1}$ regarding $v_1$ and the similarities of their commands.

\vspace{1ex}

\noindent The same construction with a maximum recursion depth of $2^n + 1$ can be used to simulate a counter program with counters bounded by $2^{2^{2^n}}$: The procedures \texttt{inc} and \texttt{dec} both work by iteratively squaring a base value of $2$ and performing that many increments or decrements. Starting with $2$ and squaring $n$-times yields $2^{2^n}$, therefore squaring $2^n$ times instead yields $2^{2^{2^n}}$.  To give a more visual explanation:
\begin{align*}
  (\cdots(2\overbrace{^2)^2\cdots)^2}^{n\text{-times}} &= 2^{2^n}\\
  (\cdots(2\overbrace{^2)^2\cdots)^2}^{2^n\text{-times}} &= 2^{2^{2^n}}.
\end{align*}

\subsection{Proofs for \cref{sec:RNPtoTPN}} \label{sec:TPNconstruction}

Given an $\rnp$ $R$ with maximum recursion depth $k$ we construct a $\TDPN$ $\mathcal{N} = (w_{\mathit{init}},$ $w_{\mathit{final}},$ $\mathcal{T}_{\mathit{move}},$ $\mathcal{T}_{\mathit{fork}},$ $\mathcal{T}_{\mathit{join}})$, which defines the Petri net $N(\mathcal{N}) = (P,T,F,p_0,p_f)$, such that $\multi{p_f}$ is coverable in $N(\mathcal{N})$ iff there is a terminating execution of $R$. Let us give more details on how to construct $N(\mathcal{N})$ and the three transducers that define it.

The main idea for $N(\mathcal{N})$ is to have up to $k+1$ many places for each variable and each label of $R$. It is clear from the semantics of recursive net programs, why we need $k+1$ places for each variable. Furthermore, we need $k-1$ places for each label appearing in the $\text{\texttt{proc}}_{<\text{max}}$-specification for a procedure \texttt{proc} to simulate the call stack. A single place per label would allow us to store the contents of the call stack, but not the order, meaning we would not be able to distinguish between some configurations. For $\text{\texttt{proc}}_{=\text{max}}$ as well as the main program, only a single place is needed per label, since the corresponding commands are only executed at recursion depths $k$ and $0$, respectively.

Transitions are introduced in similar fashion as places. We follow along the Lipton construction, with the alteration of having to copy some transitions up to $k$-times, to connect places at different recursion depths. How to simulate each command of $R$ with Petri net transitions can be inferred from Figure~\ref{fig:netcommands} and Figure~\ref{fig:netprocedures}. As we can see, each \textbf{call} command and each procedure requires an additional place per recursion depth $d$ (named $l_{1,d}$\_calls\_\texttt{proc} respectively return\_\texttt{proc}$_{d+1}$ in Figure~\ref{fig:netprocedures}). Furthermore, we need a single place $w_{\mathit{halt}}$ for the \textbf{halt} command, which we use as our $w_{\mathit{final}}$.

To make procedure calls work correctly, we also identify the starting labels of $\text{\texttt{proc}}_{<\text{max}}$ and $\text{\texttt{proc}}_{=\text{max}}$ for each procedure \texttt{proc}. For such a label $l$, the places indexed $1$ to $k-1$ then correspond to $\text{\texttt{proc}}_{<\text{max}}$, whereas the place with index $k$ corresponds to $\text{\texttt{proc}}_{=\text{max}}$. The places for labels of the main program are all indexed with $0$. For $w_{\mathit{init}}$ we use the place corresponding to the first label in the main program.

\vspace{1ex}

\noindent Let us now construct the transducers. First, we define our alphabet to be $\Sigma := \{0,1\}$. Thus, we need to give every place of $N(\mathcal{N})$ a binary address. To this end, we first have to count the number $h$ of places without counting additional copies at different recursion depths. Given a procedure \texttt{proc}, let $\#p(\text{main})$, $\#p(\text{\texttt{proc}}_{=\text{max}})$, and $\#p(\text{\texttt{proc}}_{<\text{max},d})$ be the number of places needed for the main program, for $\text{\texttt{proc}}_{=\text{max}}$, and for $\text{\texttt{proc}}_{<\text{max}}$ at recursion depth $d$, respectively. We define 
\[h := \#p(\text{main}) + \sum_{\text{\texttt{proc}} \in \mathsf{PROC}} \big(\#p(\text{\texttt{proc}}_{=\text{max}}) + \#p(\text{\texttt{proc}}_{<\text{max},1}) - 1\big).\]
Here we have to subtract $1$ for each procedure \texttt{proc}, because we identified the starting labels of $\text{\texttt{proc}}_{<\text{max}}$ and $\text{\texttt{proc}}_{=\text{max}}$. Since every labelled command results in at most $2$ places counted this way (one for the label and up to one auxiliary place), $h$ is linear in the size of $R$. Now we can give each of the counted places a different address using $\log{h}$ bits. Then for the actual places, we just append to this address the binary representation of the recursion depth $d$, that a given place corresponds to. This results in addresses of length at most $\log{h} + \log{k} = \log(h \cdot k)$. If we encode the numbers for label and recursion depth with leading zeros, all addresses also have the same length, as required.

Now we start describing the components of the ternary transducers, leaving the binary one for later. Regarding Figure~\ref{fig:netcommands}, to check whether three places are connected via a transition, most of the information is confined to the first $\log{h}$ bits, while the last $\log{k}$ bits just have to be checked for equality. Furthermore, there are only $2^{3\log{h}} = 8h$ many possibilities for triples of addresses of length $\log{h}$, which is polynomial in the size of $R$. To differentiate between all these possibilities, we can just use a ternary transducer $\mathcal{T}_{\mathit{diff}}$ with states $Q_{\mathit{diff}} = \bigcup_{j=0}^{\log{h}} \left(\{w \in \Sigma^*~:~|w|=j\}^3\right)$, initial state $(\varepsilon,\varepsilon,\varepsilon)$, and transitions $(w_1,w_2,w_3) \xrightarrow{(a_1,a_2,a_3)} (w_1.a_1,w_2.a_2,w_3.a_3)$ for all $(a_1,a_2,a_3) \in \Sigma^3$ and $(w_1,w_2,w_3) \in Q_{\mathit{diff}}$ with $|w_1| = |w_2| = |w_3| < \log{h}$. This transducer then has polynomially many states and transitions in the size of $R$. For the equality check on the last $\log{k}$ bits, we use a ternary transducer $\mathcal{T}_{\mathit{eq}}$ with states $Q_{\mathit{eq}} = \{q_0,\ldots,q_{\log{k}}\}$, initial state $q_0$, final states $\{q_{\log{k}}\}$, and transitions $q_i \xrightarrow{(a,a,a)} q_{i+1}$ for each $a \in \Sigma$ and all $i \in \{0,\ldots,\log{k}-1\}$. Since $k$ was encoded in binary for $R$, meaning it needed $\log{k}$ space, this transducer's size is also polynomial in the size of $R$.

For the transitions connecting three places in Figure~\ref{fig:netprocedures}, we can reuse the transducer $\mathcal{T}_{\mathit{diff}}$ for the first $\log{h}$ bits, but need something different from $\mathcal{T}_{\mathit{eq}}$ for the last $\log{k}$ bits. For transitions like the one connecting $l_{1,d}$, $l_{3,d+1}$ and $l_{1,d}$\_calls\_\texttt{proc}, or the one connecting $l_{1,d}$\_calls\_\texttt{proc}, return\_\texttt{proc}$_{d+1}$ and $l_{2,d}$ the second place belongs to a recursion depth one higher than the other two. Notice that we chose to order the places such a way, that it is always the second place that is different from the others. This order also complies with the definitions of $\mathcal{T}_{\mathit{fork}}$ and $\mathcal{T}_{\mathit{join}})$. To check that the second number encoded in a triple of $\log{k}$ bits is exactly one higher than the other two, we make use of a ternary transducer $\mathcal{T}_{\mathit{inc}}$ with states $Q_{\mathit{inc}} = \{0,\ldots,\log{k}\} \times \{0,1\}$, initial state $(0,0)$, final states $\{(\log{k},1)\}$, and the following transitions:
\begin{itemize}
\item $(i,0) \xrightarrow{(a,a,a)} (i+1,0)$ for all $i \in \{0,\ldots,\log{k}-1\}, a \in \Sigma$,
\item $(i,0) \xrightarrow{(0,1,0)} (i+1,1)$ for all $i \in \{0,\ldots,\log{k}-1\}$, and
\item $(i,1) \xrightarrow{(1,0,1)} (i+1,1)$ for all $i \in \{0,\ldots,\log{k}-1\}$.
\end{itemize}
This transducer suffices for all transitions connecting three places in Figure~\ref{fig:netprocedures}, since the two we mentioned already are the only ones. It is also of polynomial size in the size of $R$, since it consists at most twice as many states and transitions as $\mathcal{T}_{\mathit{eq}}$.

To now construct $\mathcal{T}_{\mathit{fork}}$, we first take a copy of $\mathcal{T}_{\mathit{diff}}$, $\mathcal{T}_{\mathit{eq}}$, and $\mathcal{T}_{\mathit{inc}}$ and use the state $(\varepsilon,\varepsilon,\varepsilon)$ of $\mathcal{T}_{\mathit{diff}}$ as the new initial state. Then we consider all states of $\mathcal{T}_{\mathit{diff}}$ that correspond to a triple which could be connected via a transition of $N$ of the appropriate form required for $\mathcal{T}_{\mathit{fork}}$ (see Figure~\ref{fig:PNtransducers}). Then depending on whether this transition of $N(\mathcal{N})$ would belong to Figure~\ref{fig:netcommands} or Figure~\ref{fig:netprocedures}, we connect the considered state to either $\mathcal{T}_{\mathit{eq}}$, or $\mathcal{T}_{\mathit{inc}}$ in the following way: Let $(w_1,w_2,w_3)$ be the considered state and let $\mathcal{T}$ be the transducer we want to connect it to. Then for every transition $q \xrightarrow{(a,b,c)} q'$ of $\mathcal{T}$, where $q$ was the initial state of $\mathcal{T}$, we add a transition $(w_1,w_2,w_3) \xrightarrow{(a,b,c)} q'$. The transducer $\mathcal{T}_{\mathit{join}}$ is constructed in the same way.

The construction of $\mathcal{T}_{\mathit{move}}$ is also very similar: We take the analogous binary transducers of $\mathcal{T}_{\mathit{diff}}$ and $\mathcal{T}_{\mathit{eq}}$ and connect them in the same way as before. This suffices, because all transitions connecting exactly two places in Figure~\ref{fig:netcommands} and Figure~\ref{fig:netprocedures} only connect places of the same recursion depth.

\vspace{1ex}

\noindent All three transducers of $\mathcal{N}$ have size polynomial in the size of $R$, and the binary addresses $w_{\mathit{init}}$ and $w_{\mathit{halt}}$ are also of length polynomial in this size. Therefore the input to the coverability problem for $\mathcal{N}$ fulfils the size requirements. Furthermore, we can see that this construction can be done in polynomial time: We just construct $N(\mathcal{N})$ without additional copies of places for different recursion depths and then construct the three transducers from there.

\section{Reducing Non-Inheritance to Inheritance $\DCPS$} \label{sec:inheritance}

In \cite{AtigBQ2009} the authors consider a variant of $\DCPS$ that has a slightly changed relation $\rightarrow_i$ compared to ours: Each newly spawned thread starts with $i+1$ as its context switch number instead of $0$, where $i$ is the context switch number of the thread that spawned it. Formally, for all $i \in \mathbb{N}, w \in \Gamma^*$ and each rule $g|\gamma \hookrightarrow g'|w' \triangleright \gamma'$ we have $\langle g, (\gamma.w,i), \mmap \rangle \rightarrow_i \langle g', (w'.w,i), \mmap' \rangle$, where now $\mmap' = \mmap \oplus [(\gamma',i+1)]$ instead of $\mmap \oplus [(\gamma',0)]$. We call this model $\DCPS$ \emph{with inheritance}, because each thread basically inherits the context switches from its parent. Our original model could then also be referred to as $\DCPS$ \emph{without inheritance}.

To make use of the $\TWOEXPSPACE$-membership result from \cite{AtigBQ2009}, let us reduce $\SRP[K]$ for $\DCPS$ without inheritance to $\SRP[K+2]$ for $\DCPS$ with inheritance.

\vspace{1ex}

\noindent Let $\mathcal{A} = (G,\Gamma,\Delta,g_0,\gamma_0)$ be a $\DCPS$ without inheritance and let $g_{\mathit{reach}} \in G$, $K \in \mathbb{N}$. We construct a $\DCPS$ with inheritance $\mathcal{A}' = (G',\Gamma',\Delta',g_0',\gamma_0')$ with $g'_{\mathit{reach}} \in G'$ such that $g_{\mathit{reach}}$ is $K$-bounded reachable in $\mathcal{A}$ iff $g'_{\mathit{reach}}$ is $(K + 2)$-bounded reachable in $\mathcal{A}'$:
\begin{itemize}
\item $G' := \{g'_0\} \cup (G \times \{0,1,2\}) \cup (G \times \Gamma)$,
\item $g'_{\mathit{reach}} = (g_{\mathit{reach}},2)$,
\item $\Gamma' := \{\gamma'_0,\gamma_{\mathit{dorm}},\gamma_{\mathit{step}},\bot,\top\} \cup \Gamma \cup \bar{\Gamma}$, where $\bar{\Gamma} = \{\bar{\gamma}|\gamma \in \Gamma\}$, and
\item $\Delta'$ contains the following transition rules:
  \begin{enumerate} 
  \item $g'_0|\gamma'_0 \hookrightarrow g'_0|\gamma'_0 \triangleright \gamma_{\mathit{dorm}}$,
  \item $g'_0|\gamma'_0 \hookrightarrow g'_0|\varepsilon \triangleright \bot$,
  \item $g'_0|\bot \hookrightarrow (g_0,0)|\gamma_0.\bot$
  \item $(g_1,0)|\gamma \hookrightarrow (g_2,0)|w$ iff $g_1|\gamma \hookrightarrow g_2|w \in \Delta$,
  \item $(g_1,0)|\gamma_1 \hookrightarrow (g_2,0)|w \triangleright \bar{\gamma_2}$ iff $g_1|\gamma_1 \hookrightarrow g_2|w \triangleright \gamma_2 \in \Delta$,
  \item $(g,0)|\gamma \hookrightarrow (g,1)|\top.\gamma \triangleright \gamma_{\mathit{step}}$ for each $g \in G$ and each $\gamma \in \Gamma \cup \{\bot\}$,
  \item $(g,1)|\gamma_{\mathit{step}} \hookrightarrow (g,2)|\varepsilon$ for each $g \in G$,
  \item $(g,2)|\top \hookrightarrow (g,0)|\varepsilon$ for each $g \in G$,
  \item $(g,2)|\bar{\gamma} \hookrightarrow (g,\gamma)|\varepsilon$ for each $g \in G$ and each $\gamma \in \Gamma$, and
  \item $(g,\gamma)|\gamma_{\mathit{dorm}} \hookrightarrow (g,0)|\gamma.\bot$ for each $g \in G$ and each $\gamma \in \Gamma$.
  \end{enumerate}
\end{itemize}
To make it clear which rule is being applied for each $\rightarrow_i$-related pair of configurations, we put the rule number above the arrow.

We prove the following statements:
\begin{enumerate}
\item If $\langle g, (w,i), \mmap \rangle$ is $K$-bounded reachable in $\mathcal{A}$ then $\langle (g,0), (w.\bot,i+1), \mmap' \rangle$ is $(K+2)$-bounded reachable in $\mathcal{A}'$, where
  \begin{itemize}
  \item $\mmap(\gamma,0) = \sum_{j'=0}^{K+2} \mmap'(\bar{\gamma},j')$ for all $\gamma \in \Gamma$,
  \item $\mmap(v,j) = \mmap'(\top.v.\bot,j+1)$ for all $(v,j) \in \Gamma^* \times \{1,\ldots,K\}$, and
  \item $\mmap'(v',j') \neq 0$ for any $j' \in \mathbb{N}$ implies that $v' = \varepsilon$, $v' \in \bar{\Gamma}$, or $v' = \top.v.\bot$ for some $v \in \Gamma^*$.
  \end{itemize}
\item If $\langle (g,0), (w.\bot,i'), \mmap' \rangle$ with $i' \leq K+1$ is $(K+2)$-bounded reachable in $\mathcal{A}'$ then $\langle g, (w,i),$ $\mmap \rangle$ is $K$-bounded reachable in $\mathcal{A}$, where
  \begin{itemize}
  \item $i+1 \leq i'$,
  \item $\mmap$ and $\mmap'$ are related in the same way as in the previous statement, except every context switch number of a local configuration in $\mmap'$ is allowed to be arbitrarily higher and $\mmap'$ is allowed to contain additional local configurations with stack content $\gamma_{dorm}$.
  \end{itemize}
\end{enumerate}

\noindent Regarding the first statement, we use induction on the length of the sequence of $\Rightarrow_{\leq K}$-related configurations of $\mathcal{A}$. In the base case, this sequence consists of just the initial configuration $\langle g_0, (\gamma_0,0), \emptyset \rangle$. Consider the following configuration sequence of $\mathcal{A}'$:
\begin{align*}
  &\langle g'_0, (\gamma'_0,0), \emptyset \rangle \\
  &\xrightarrow{2}_{0} \langle g'_0, (\varepsilon,0), \multi{(\bot,1)} \rangle \\
  &\mapsto_{0} \langle g'_0, (\bot,1), \multi{(\varepsilon,1)} \rangle \\
  &\xrightarrow{3}_{1} \langle (g_0,0), (\gamma_0.\bot,1), \multi{(\varepsilon,1)} \rangle.
\end{align*}
The last configuration in this sequence fulfils the requirement.

For the inductive case, assume there is a $K$-bounded reachable configuration $C$ of $\mathcal{A}$, for which the statement already holds. Let $C'$ be a configuration of $\mathcal{A}$ with $C \Rightarrow_{\leq K} C'$. We proceed by going over all cases for why this pair could be related via $\Rightarrow_{\leq K}$:
\begin{description}
\item[{Case $\langle g, (\gamma.w,i), \mmap \rangle \rightarrow_i \langle g', (w'.w,i), \mmap \rangle$:}]\ \\
By induction hypothesis, a configuration of the form
\[\langle (g,0), (\gamma.w.\bot,i+1), \mmap' \rangle\]
is $(K+2)$-bounded reachable in $\mathcal{A}'$, where $\mmap'$ corresponds to $\mmap$ in the required way. Furthermore, $\Delta$ has to contain the rule $g|\gamma \hookrightarrow g'|w'$ for the above relation to hold. Therefore, $\Delta'$ contains the rule $(g,0)|\gamma \hookrightarrow (g',0)|w'$ by definition. Thus, we have
\[\langle (g,0), (\gamma.w.\bot,i+1), \mmap' \rangle \xrightarrow{4}_{i+1} \langle (g,0), (w'.w.\bot,i+1), \mmap' \rangle.\]
The last configuration here fulfils the requirement.
\vspace{1ex}
\item[{Case $\langle g, (\gamma.w,i), \mmap \rangle \rightarrow_i \langle g', (w'.w,i), \mmap \oplus \multi{(\gamma',0)} \rangle$:}]\ \\
By induction hypothesis, a configuration of the form
\[\langle (g,0), (\gamma.w.\bot,i+1), \mmap' \rangle\]
is $(K+2)$-bounded reachable in $\mathcal{A}'$, where $\mmap'$ corresponds to $\mmap$ in the required way. Furthermore, $\Delta$ has to contain the rule $g|\gamma \hookrightarrow g'|w' \triangleright \gamma'$ for the above relation to hold. Therefore, $\Delta'$ contains the rule $(g,0)|\gamma \hookrightarrow (g',0)|w' \triangleright \bar{\gamma'}$ by definition. Thus, we have
\[\langle (g,0), (\gamma.w.\bot,i+1), \mmap' \rangle \xrightarrow{5}_{i+1} \langle (g,0), (w'.w.\bot,i+1), \mmap' \oplus \multi{(\bar{\gamma'},i+2)} \rangle.\]
The last configuration here fulfils the requirement, because $i \leq K$ due to adhering to the bound in $\mathcal{A}$.
\vspace{1ex}
\item[{Case $\langle g, (w,i), \mmap \oplus \multi{(w',0)} \rangle \mapsto_i \langle g, (w',0), \mmap \oplus \multi{(w,i+1)} \rangle$:}]\ \\
Since inactive threads with context switch number $0$ have never been active, their stack contents cannot have changed since they were spawned. Therefore $|w'| = 0$, or rather $w' = \gamma \in \Gamma$. By induction hypothesis, a configuration of the form
\[\langle (g,0), (w.\bot,i+1), \mmap' \oplus \multi{(\bar{\gamma},j')} \rangle\]
is $(K+2)$-bounded reachable in $\mathcal{A}'$, where $j' \leq K+2$ and $\mmap'$ corresponds to $\mmap$ in the required way. To the configuration sequence of $\mathcal{A}'$, that serves as a witness for this reachability, we add the following at the start:
\[\langle g'_0, (\gamma'_0,0), \emptyset \rangle \xrightarrow{1}_{0} \langle g'_0, (\gamma'_0,0), \multi{(\gamma_{\mathit{dorm}},1)} \rangle.\]
In the process, every configuration in the sequence receives an additional inactive thread with local configuration $(\gamma_{\mathit{dorm}},1)$. Thus, the configuration 
\[\langle (g,0), (w.\bot,i+1), \mmap' \oplus \multi{(\bar{\gamma},j')} \oplus \multi{(\gamma_{\mathit{dorm}},1)} \rangle\]
is also $(K+2)$-bounded reachable in $\mathcal{A}'$. From there, we consider the following configuration sequence:
\begin{align*}
  &\langle (g,0), (w.\bot,i+1), \mmap' \oplus \multi{(\bar{\gamma},j')} \oplus \multi{(\gamma_{\mathit{dorm}},1)} \rangle\\
  &\xrightarrow{6}_{i+1} \langle (g,1), (\top.w.\bot,i+1), \mmap' \oplus \multi{(\bar{\gamma},j')} \oplus \multi{(\gamma_{\mathit{dorm}},1)} \oplus \multi{(\gamma_{\mathit{step}},i+2)} \rangle\\
  &\mapsto_{i+1} \langle (g,1), (\gamma_{\mathit{step}},i+2), \mmap' \oplus \multi{(\bar{\gamma},j')} \oplus \multi{(\gamma_{\mathit{dorm}},1)} \oplus \multi{(\top.w.\bot,i+2)} \rangle\\
  &\xrightarrow{7}_{i+2} \langle (g,2), (\varepsilon,i+2), \mmap' \oplus \multi{(\bar{\gamma},j')} \oplus \multi{(\gamma_{\mathit{dorm}},1)} \oplus \multi{(\top.w.\bot,i+2)} \rangle\\
  &\mapsto_{i+2} \langle (g,2), (\bar{\gamma},j'), \mmap' \oplus \multi{(\top.w.\bot,i+2)} \oplus \multi{(\gamma_{\mathit{dorm}},1)} \oplus \multi{(\varepsilon,i+3)} \rangle\\
  &\xrightarrow{9}_{j'} \langle (g,\gamma), (\varepsilon,j'), \mmap' \oplus \multi{(\top.w.\bot,i+2)} \oplus \multi{(\gamma_{\mathit{dorm}},1)} \oplus \multi{(\varepsilon,i+3)} \rangle\\
  &\mapsto_{j'} \langle (g,\gamma), (\gamma_{\mathit{dorm}},1), \mmap' \oplus \multi{(\top.w.\bot,i+2)} \oplus \multi{(\varepsilon,j'+1)} \oplus \multi{(\varepsilon,i+3)} \rangle\\
  &\xrightarrow{10}_{1} \langle (g,0), (\gamma.\bot,1), \mmap' \oplus \multi{(\top.w.\bot,i+2)} \oplus \multi{(\varepsilon,j'+1)} \oplus \multi{(\varepsilon,i+3)} \rangle.
\end{align*}
The relations here hold, since $i \leq K$ due to adhering to the bound $K$ in $\mathcal{A}$. The last configuration in this sequence fulfils the requirement.
\vspace{1ex}
\item[{Case $\langle g, (w,i), \mmap \oplus \multi{(w',j)} \rangle \mapsto_i \langle g, (w',j), \mmap \oplus \multi{(w,i+1)} \rangle$, where $j > 0$:}]\ \\
By induction hypothesis, a configuration of the form
\[\langle (g,0), (w.\bot,i+1), \mmap' \oplus \multi{(\top.w'.\bot,j+1)} \rangle\]
is $(K+2)$-bounded reachable in $\mathcal{A}'$, where $\mmap'$ corresponds to $\mmap$ in the required way. From here, we consider the following configuration sequence of $\mathcal{A}'$:
\begin{align*}
  &\langle (g,0), (w.\bot,i+1), \mmap' \oplus \multi{(\top.w'.\bot,j+1)} \rangle\\
  &\xrightarrow{6}_{i+1} \langle (g,1), (\top.w.\bot,i+1), \mmap' \oplus \multi{(\top.w'.\bot,j+1)} \oplus \multi{(\gamma_{\mathit{step}},i+2)} \rangle\\
  &\mapsto_{i+1} \langle (g,1), (\gamma_{\mathit{step}},i+2), \mmap' \oplus \multi{(\top.w'.\bot,j+1)} \oplus \multi{(\top.w.\bot,i+2)} \rangle\\
  &\xrightarrow{7}_{i+2} \langle (g,2), (\varepsilon,i+2), \mmap' \oplus \multi{(\top.w'.\bot,j+1)} \oplus \multi{(\top.w.\bot,i+2)} \rangle\\
  &\mapsto_{i+2} \langle (g,2), (\top.w'.\bot,j+1), \mmap' \oplus \multi{(\varepsilon,i+3)} \oplus \multi{(\top.w.\bot,i+2)} \rangle\\
  &\xrightarrow{8}_{j+1} \langle (g,0), (w'.\bot,j+1), \mmap' \oplus \multi{(\varepsilon,i+3)} \oplus \multi{(\top.w.\bot,i+2)} \rangle
\end{align*}
The relations here hold, since $i \leq K$ due to adhering to the bound $K$ in $\mathcal{A}$. The last configuration in this sequence fulfils the requirement.
\end{description}
This concludes the proof of the first statement.

\vspace{1ex}

\noindent Regarding the second statement, we use induction on the number of global states in $G \times \{0\}$, that appear in the sequence of $\Rightarrow_{\leq K + 2}$-related configurations of $\mathcal{A}'$. For the base case, we again consider the following configuration sequence:
\begin{align*}
  &\langle g'_0, (\gamma'_0,0), \emptyset \rangle \\
  &\xrightarrow{2}_{0} \langle g'_0, (\varepsilon,0), \multi{(\bot,1)} \rangle \\
  &\mapsto_{0} \langle g'_0, (\bot,1), \multi{(\varepsilon,1)} \rangle \\
  &\xrightarrow{3}_{1} \langle (g_0,0), (\gamma_0.\bot,1), \multi{(\varepsilon,1)} \rangle.
\end{align*}
The last configuration in this sequence corresponds to the initial configuration of $\mathcal{A}$ as required. The only way to deviate from this sequence without including more than one configuration of $G \times \{0\}$ is to spawn additional threads with stack content $\gamma_{\mathit{dorm}}$ and to perform arbitrary context switches between these.

For the inductive case, we assume that we have $C \Rightarrow^*_{\leq K + 2} C'$ for configurations $C,C'$ of $\mathcal{A}'$, whose global states are in $G \times \{0\}$. Furthermore, we assume that the statement already holds for $C$ and that the $\Rightarrow_{\leq K + 2}$-related sequence from $C$ to $C'$ contains no further configurations, whose global states are in $G \times \{0\}$. For the most part, it is easy to see, that such a sequence has to fall under one of the cases considered in the proof of the first statement, but may contain additional context switches. In all these cases, we end up with a configuration $C'$, for which the second statement holds.

The one case, that we still need to consider, is having an active thread with local configuration $(w.\bot,K+2)$ and a global state in $G \times \{0\}$. This case is not covered by the second statement and therefore we should proof that it does not cause any problems. However, let us first argue that this case can actually occur. Consider the following sequence:
\begin{align*}
  &C = \langle (g,0), (w.\bot,i'), \mmap' \oplus \multi{(\top.w'.\bot,K+2)} \rangle\\
  &\xrightarrow{6}_{i'} \langle (g,1), (\top.w.\bot,i'), \mmap' \oplus \multi{(\top.w'.\bot,K+2)} \oplus \multi{(\gamma_{\mathit{step}},i'+1)} \rangle\\
  &\mapsto_{i'} \langle (g,1), (\gamma_{\mathit{step}},i'+1), \mmap' \oplus \multi{(\top.w'.\bot,K+2)} \oplus \multi{(\top.w.\bot,i'+1)} \rangle\\
  &\xrightarrow{7}_{i'+1} \langle (g,2), (\varepsilon,i'+1), \mmap' \oplus \multi{(\top.w'.\bot,K+2)} \oplus \multi{(\top.w.\bot,i'+1)} \rangle\\
  &\mapsto_{i'+1} \langle (g,2), (\top.w'.\bot,K+2), \mmap' \oplus \multi{(\varepsilon,i'+1)} \oplus \multi{(\top.w.\bot,i'+1)} \rangle\\
  &\xrightarrow{8}_{K+2} \langle (g,0), (w'.\bot,K+2), \mmap' \oplus \multi{(\varepsilon,i'+2)} \oplus \multi{(\top.w.\bot,i'+1)} \rangle
\end{align*}
We can see that this case can occur any time we have an inactive thread with local configuration $(\top.w'.\bot,K+2)$ for some $w' \in \Gamma^*$. Let us argue that from here $\mathcal{A}'$ can no longer $(K+2)$-bounded reach a configuration of the form $\langle (\tilde{g},0), (\tilde{w}.\bot,\tilde{i}), \tilde{\mmap} \rangle$, where $\tilde{w} \in \Gamma^*$ and $\tilde{i} < K+2$. Considering $C$ matched the requirements of the second statement, the only way to reach such a configuration is to go to some state in $G \times \{2\}$ first, to then switch in an inactive thread with local configuration $(\top.\tilde{w}.\bot,\tilde{i})$ and afterwards pop the $\top$-symbol. An attempt at this leads to the following sequence:
\begin{align*}
  &\langle (g,0), (w'.\bot,K+2), \mmap' \oplus \multi{(\varepsilon,i'+2)} \oplus \multi{(\top.w.\bot,i'+1)} \rangle\\
  &\xrightarrow{6}_{K+2} \langle (g,1), (\top.w'.\bot,K+2), \\
  &\qquad\qquad\qquad \mmap' \oplus \multi{(\varepsilon,i'+2)} \oplus \multi{(\top.w.\bot,i'+1)} \oplus \multi{(\gamma_{\mathit{step}},K+3)} \rangle
\end{align*}
Here, the inactive thread with local configuration $(\gamma_{\mathit{step}},K+3)$ cannot be switched in, since its context switch number is already too high. The multiset $\mmap'$ also contains no other $\gamma_{\mathit{step}}$-threads by virtue of $C$ matching the requirements of the second statement. Thus, $\mathcal{A}'$ can no longer change the global state from here. This concludes the proof of the second statement.

\vspace{1ex}

\noindent Finally, let us prove that $g_{\mathit{reach}}$ is $K$-bounded reachable in $\mathcal{A}$ iff $(g_{\mathit{reach}},2)$ is $(K + 2)$-bounded reachable in $\mathcal{A}'$:

If $g_{\mathit{reach}}$ is $K$-bounded reachable in $\mathcal{A}$, then a configuration of the form $\langle (g,0), (w.\bot,i+1), \mmap' \rangle$ is $(K + 2)$-bounded reachable in $\mathcal{A}'$ by the first statement. Consider the following configuration sequence of $\mathcal{A}'$, which reaches $(g_{\mathit{reach}},2)$:
\begin{align*}
  &\langle (g_{\mathit{reach}},0), (w.\bot,i+1), \mmap' \rangle\\
  &\xrightarrow{6}_{i+1} \langle (g_{\mathit{reach}},1), (\top.w.\bot,i+1), \mmap' \oplus \multi{(\gamma_{step},i+2)} \rangle\\
  &\mapsto_{i+1} \langle (g_{\mathit{reach}},1), (\gamma_{step},i+2), \mmap' \oplus \multi{(\top.w.\bot,i+2)} \rangle\\
  &\xrightarrow{7}_{i+2} \langle (g_{\mathit{reach}},2), (\varepsilon,i+2), \mmap' \oplus \multi{(\top.w.\bot,i+2)} \rangle
\end{align*}

On the other hand, if $(g_{\mathit{reach}},2)$ is $(K + 2)$-bounded reachable in $\mathcal{A}'$ then $(g_{\mathit{reach}},0)$ must have been reachable as well, because the transition rules of $\mathcal{A}'$ only allow the $G$-part of the state tuple to change while in $G \times \{1\}$. Furthermore, only top of the stack symbols from $\Gamma$ allow for changes in the $G$-part of the state tuple, and stacks with such symbols always receive a bottom of the stack symbol $\bot$ according to the transition rules of $\mathcal{A}'$. Additionally, an active thread with such stack contents must have a context switch number $< K+2$ to reach $G \times \{2\}$ from $G \times \{0\}$, as we have seen previously. Finally, we can assume there to be no inactive threads with stack content $\gamma_{step}$, since such threads only exist in $G \times \{1\}$ and get consumed to move away from there.

Therefore, a configuration of the form $\langle (g_{\mathit{reach}},0), (w.\bot,i'), \mmap' \rangle$ as required by the second statement is $(K + 2)$-bounded reachable in $\mathcal{A}'$. This means that a configuration of the form $\langle g_{\mathit{reach}}, (w,i), \mmap \rangle$ with $i \leq i'-1$ is $K$-bounded reachable in $\mathcal{A}$, and so is $g_{reach}$ itself.

\end{document}